\documentclass[12pt]{article}
\usepackage{fullpage, float}

\usepackage[T1]{fontenc}

\usepackage{fancyhdr}
\usepackage{amsmath, color, epsfig,bbm,outline}
\usepackage{amsthm}
\usepackage{amscd, amssymb}
\usepackage[all]{xy}
\newtheorem{thm}{Theorem}[section]
\newtheorem{theorem}{Theorem}[section]
\newtheorem*{thm*}{Theorem}

\newtheorem*{cor*}{Corollary}

\newtheorem{lem*}{Lemma}

\newtheorem{proposition}[thm]{Proposition}

\newtheorem*{con*}{Conjecture}
\newtheorem*{prob*}{Problem}

\theoremstyle{definition} 

\newtheorem{definition}{Definition}

\newtheorem{assumption}{Assumption}
\newtheorem{example}{Example}

\newtheorem{remark}[thm]{Remark}

\theoremstyle{remark}

\newcommand{\bbC}{\mathbb{C}}

\newcommand{\bbE}{\mathbb{E}}

\newcommand{\bbP}{\mathbb{P}}
\newcommand{\bbQ}{\mathbb{Q}}
\newcommand{\bbR}{\mathbb{R}}

\newcommand{\m}[1]{{\mathcal{#1}}}

\providecommand{\abs}[1]{\lvert#1\rvert}

\newcommand{\comment}[1]{}

\newcounter{exercise}

{\end{itshape}\vspace{10pt}\par}

\usepackage{natbib}

\title{On a dynamic adaptation of the Distribution Builder approach to investment decisions}
\author{Phillip Monin\footnote{Department of Mathematics, The University of Texas at Austin, Austin, TX 78712. Email: pmonin@math.utexas.edu.}}
\date{First draft: June 2011 \\ This draft: November 2012}

\begin{document}

\maketitle
\abstract{

Sharpe \emph{et al.} proposed the idea of having an expected utility maximizer choose a probability distribution for future wealth as an input to her investment problem instead of a utility function. They developed a computer program, called \emph{The Distribution Builder}, as one way to elicit such a distribution. In a single-period model, they then showed how this desired distribution for terminal wealth can be used to infer the investor's risk preferences. We adapt their idea, namely that a risk-averse investor can choose a desired distribution for future wealth as an alternative input attribute for investment decisions, to continuous time. In a variety of scenarios, we show how the investor's desired distribution combines with her initial wealth and market-related input to determine the feasibility of her distribution, her implied risk preferences, and her optimal policies throughout her investment horizon. We then provide several examples.\\

}

\section{Introduction}

Theoretical models in single-agent investment are traditionally based on the classical criterion of maximal expected utility of wealth. Despite its long history and sound economic foundations, however, expected utility as a criterion for practical investment choice 
faces many obstacles due to various difficulties for its specification. Some of these difficulties have been addressed by making simplifying or ad hoc assumptions.  Asset managers, for instance, often make two such assumptions. First, they assume that the investor has constant relative risk aversion. They then use so-called risk tolerance quizzes to approximate the investor's relative risk aversion coefficient. 

Alternatively, one can focus on observable features of investors' behavior. For instance, \citet{bla-1968}, among others, proposed to essentially bypass the utility concept altogether and, instead, 
use the investor's initial choice of optimal investment as the criterion to determine future optimal allocations. In a related direction, several papers have studied the specification of utility if one knows \emph{a priori} the optimal allocations that are consistent with this utility 
(see, among others, \citet{cox-lel-2000}, \citet{he-hua-1994}, \citet{dyb-rog-1997} and \citet{cox-hob-obl-2012}). 

Sharpe and his collaborators took a different point of view in \citet{sha-gol-bly-2000}, \citet{sha-2001}, and \citet{gol-joh-sha-2008}. 
They argued that, in practical situations, investors can express desires about the distribution of their future wealth. 
To gather such distributional data, they developed a computer program, called \emph{The Distribution Builder}, whose output is a probability 
distribution that the investor desires for her future wealth. Then, in a single-period model and under the assumption that the investor implicitly maximizes her 
expected utility of terminal wealth, Sharpe \emph{et al.} showed how this desired distribution can be used to recover the investor's risk 
preferences.

Our work is inspired and motivated by this approach. The aim herein is to provide a dynamic adaptation of their idea, which is to use a risk-averse investor's desired distribution for future wealth, rather than a utility function, as an input for optimal investment. Given an investor's desired distribution for future wealth and her initial endowment, we study the following issues: if this distribution can be achieved in the market, how it is achieved, and, finally, the risk preferences that are consistent with this choice of distribution. As in the work of Sharpe \emph{et al.}, we address, in a practical way, both the 
normative issue of instructing investors how to achieve their goals as well as the theoretical question of how to infer risk 
preferences that are consistent with investment targets.

Given that we work beyond a single-period setting, the time at which the investor wants to achieve her desired distribution is 
an important input parameter in the analysis. We consider two scenarios. In the first, we assume that the investor implicitly 
maximizes her expected utility of terminal wealth in a \emph{fixed horizon setting}, by which we mean that the investor has a finite 
and fixed investment horizon that is specified when investment begins. Within the fixed horizon setting, we consider two subcases 
depending on whether the investor targets her distribution for terminal wealth or for wealth 
at some intermediate time. This scenario is appropriate for an investor who is certain about the length of her investment horizon and is not interested in exploring investment opportunities beyond it while she is investing. In the second scenario, we assume that the investor operates in a \emph{flexible horizon setting}, 
by which we mean that the time at which investment ends is not predetermined and could be finite or infinite. The investor places her chosen distribution for wealth at some arbitrary future time. This scenario is appropriate for an investor who does not want to commit at initial time to a fixed investment horizon, or plans to invest for a very long time.

The market environment that we consider consists of risky stocks and a riskless money market account. The stock prices are modeled as geometric Brownian motions with time-varying deterministic coefficients.

Our results are as follows. In the fixed horizon setting, we show that the desired distribution, the investor's initial wealth, and market-related input are sufficient to explicitly determine the feasibility of the investor's choice of distribution, the optimal strategy the investor should follow to attain her goal, and the investor's terminal marginal utility function. We obtain these results regardless of whether the investor targets her distribution for terminal wealth or for wealth at an intermediate time.

We obtain analogous results for the flexible horizon setting. Here, the terminal-horizon expected maximal utility criterion needs to be modified, and for this we use the so-called monotone forward investment performance criterion. Again, we show that the investor's desired distribution, her initial wealth, and market-related input are sufficient to determine the feasibility of the distribution, the strategy that achieves it, and her risk preferences.

In the fixed horizon setting, the method of proof relies on known representation results for the optimal wealth process in terms of the solution to the heat equation and on the work of Widder on inverting the Weierstrass transform. In the flexible horizon setting, it is shown that the investor's distribution, initial wealth, and market input determine the Fourier transform of a particular Borel measure that is known to characterize all objects of interest in the model under the monotone forward investment performance investment criterion.

Our results show that in our model, a desired distribution for wealth at a \emph{single} future time, when combined with the investor's initial wealth and an estimate of the market price of risk throughout the investment horizon, explicitly determines the investor's risk preferences, her optimal policies throughout, and the feasibility of her chosen distribution. This result holds regardless of whether the investor is a classical expected utility maximizer with a fixed investment horizon or whether she uses the monotone forward investment performance criterion with a flexible investment horizon.

The paper is organized as follows. In section \ref{sect: sharpe et al}, we review the method underlying \emph{The Distribution Builder}. In section \ref{sect: model and background}, we present the continuous-time model and relevant background results on the expected utility and monotone forward investment performance investment criteria. In section \ref{sect: fixed}, we consider targeted wealth distributions in the fixed horizon setting, while in section \ref{sect: flexible} we consider targeted wealth distributions in the flexible horizon setting. We provide conclusions and directions for future research in section \ref{sect: conclusions}.

\section{Single-period investment model and its Distribution Builder}\label{sect: sharpe et al}

To motivate the reader, we review the model setting and the method of \emph{The Distribution Builder} developed by Sharpe \emph{et al.} (see \citet{sha-gol-bly-2000}, \citet{sha-2001}, and \citet{gol-joh-sha-2008}). Therein, three key model assumptions were made: i) the state price density is solely expressed in terms of the stock price, ii) the investor is implicitly an expected utility maximizer, but specifies her desired future wealth distribution instead of her utility function, and iii) the investor wants to obtain her desired distribution in a so-called cost-efficient manner. We elaborate on their model and on these assumptions next.

The model is a single-period one having $N>2$ distinct possible states $\Omega:=\{\omega_i\}_{i=1}^{N}$, each occurring with equal probability $\bbP\{\omega_i\}=\frac{1}{N},\;  i=1,\dots,N$. The market consists of one riskless money market and one risky stock. The former has initial price $B_0=1$ and is assumed to offer constant interest rate $r>0$, i.e. $B_T(\omega_i)=(1+r),\; i=1,\dots,N$. 

The stock has initial price $S_0=1$ and its terminal values in the $N$ states are determined by a discrete approximation to a lognormal distribution. This is accomplished as follows. The logarithmic return of the stock is assumed to be normally distributed with mean $\mu>0$ and standard deviation $\sigma>0$. The resulting continuous distribution is then lognormal and can be approximated by selecting $N$ points with probablities $\frac{1}{2N},\frac{3}{2N},\dots,\frac{2N-1}{2N}$ from the inverse of its cumulative distribution function. This in turn produces the vector $S_T$ of $N$ equally probable states. Without loss of generality, it is assumed that the states are in nondecreasing order,
\begin{equation}\label{eq: sharpe stock increasing}
 S_T(\omega_i)\leq S_T(\omega_{i+1}),\qquad  i=1,\dots,N-1.
\end{equation}
Moreover, to preclude arbitrage in this model, the familiar assumption $S_T(\omega_1)<1+r<S_T(\omega_N)$ is introduced.

The market admits a state price density vector $\xi_T$, which is not unique because of incompleteness. Sharpe \emph{et al.} then make the ad hoc assumption that the logarithm of the vector $\xi_T$ satisfies the linear relationship
\begin{equation}\label{eq: llpw}
 \log(\xi_T(\omega_i))=a+b\log(S_T(\omega_i)),\qquad i=1,\dots, N,
\end{equation}
for some constants $a$ and $b$. To find these constants, one uses the identities
\begin{equation*}
 \frac{1}{N}\sum_{i=1}^{N}\xi_T(\omega_i) = \frac{1}{1+r}\qquad\text{and}\qquad \frac{1}{N}\sum_{i=1}^{N}\xi_T(\omega_i)S_T(\omega_i)=S_0=1,\qquad i=1,\dots,N,
\end{equation*}
to derive the equation
\begin{equation}\label{eq: sharpe nonlinear b}
 (1+r)\sum_{i=1}^{N} S_T^b(\omega_i) = \sum_{i=1}^{N}S_T^{b+1}(\omega_i).
\end{equation}
This equation then determines $b$ and using \eqref{eq: llpw} we, in turn, find $a$. It is easily shown that if $\mu>r$ then the solution $b$ to \eqref{eq: sharpe nonlinear b} exists, is unique, and is strictly negative. \comment{In turn, by \eqref{eq: sharpe stock increasing} and \eqref{eq: llpw}, it follows that the vector $\xi_T$ is nonincreasing, i.e. $\xi_T(\omega_i)\geq \xi_T(\omega_{i+1}),\;  i=1,\dots,N-1$.}

The assumption that the stock price and state price density are related as in \eqref{eq: llpw} seems at first to be restrictive and arbitrary. This relationship, however, is consistent with widely used models of asset prices, examples of which include multiperiod iid binomial models in discrete time and the classical Black-Scholes-Merton model in continuous time (see \citet{sha-2001} for further discussion).

In this market environment, the investor starts with initial wealth $x_0>0$ and sets an investment goal, namely a probability distribution denoted by $F$, for her terminal wealth. As we describe in detail below, the issue of whether $F$ can be attained depends on $x_0$ and on market-related input. To achieve an attainable distribution, the investor chooses at initial time how much money $\pi$ to allocate to the risky asset, with the remaining quantity $x_0-\pi$ invested in the money market. Her wealth at time $T$ is, then, given by the random variable (recall $S_0=1$)
\begin{equation*}
 X_T(\omega) = \pi S_T(\omega) + (x_0-\pi)(1+r).
\end{equation*}

The wealth distribution $F$ is characterized by its probability mass function, namely
\begin{equation*}
\bbP(X_T=x_i)=\frac{n_i}{N},\qquad n_i\in\{0,1,\dots,N\},\qquad i=1,\dots,N. 
\end{equation*}
Therefore, $F$ can be viewed as an $N$-vector, $\bar{X}^F=\{x_i^F\}_{i=1}^N$, of wealth values where, for each $i=1,\dots,N$, we assign $n_i$ values equal to $x_i$.  Without loss of generality, the values of $\bar{X}^F$ are assumed to be in nondecreasing order, i.e. $x^F_i\leq x^F_{i+1},\; i=1,\dots,N$.  Given this assumption, there is a one-to-one correspondence between the distribution $F$ and the wealth vector $\bar{X}^F$, in the sense that for every distribution $F$ there is a given wealth vector $\bar{X}^F$, and vice-versa.

To find a terminal wealth random variable $X_T$ with a given distribution $F$, one associates each of the $N$ values in the vector $\bar{X}^F$ with one of the $N$ states of the world. There are $N!$ possible such bijections and each has a potentially different associated cost. For fixed $j=1,\dots, N!$, let $X^j_T\colon\Omega\to \bar{X}^F$ be such a bijection. Then, the cost of the distribution $F$ attained using the random variable $X_T^j$ is found by computing the inner product $C(j)$, defined by
\begin{equation*}
C(j)=\frac{1}{N}\sum_{i=1}^N \xi_T(\omega_i)X^j_T(\omega_i).                                                                                                                                                                                                                                                                                                                                                                                                                                                                                                                                                                                                                                                                                                                                                                                                                                                                                                                                                                                                   
 \end{equation*}

Sharpe \emph{et al.} assume that the investor is implicitly choosing a distribution that maximizes her expected utility of terminal wealth. In a complete market, it is well known that the optimal strategy of an investor who maximizes expected terminal utility is \emph{cost-efficient}, i.e. it achieves the so-called \emph{distributional price}  
\begin{equation}\label{eq: sharpe distributional price}
 P_D(F):=\min_{j=1,\dots,N!} C(j)
\end{equation}
of the distribution $F$ (see \citet{dyb-1988a} and \citet{dyb-1988b}). This is not true, however, in the incomplete market herein. The optimal strategy is not necessarily cost-efficient. Nevertheless, Sharpe \emph{et al.} assume that the investor does prefer to obtain her desired distribution $F$ using a cost-efficient strategy. One can then use the results of \citet{dyb-1988a} to deduce that the strategy $j^*$, defined by
\begin{equation}\label{eq: sharpe random variable}
X^{j^*}_T(\omega_i)=x_i^F,\qquad  i=1,\dots,N, 
\end{equation}
is cost-efficient. Moreover, if $j^*$ also satisfies $C(j^*)\leq x_0$, then it corresponds to the optimal investment strategy for the investor maximizing her expected utility of terminal wealth.
 
We are now ready to review the results of Sharpe \emph{et al.} on how to infer points on the investor's marginal utility curve from her desired distribution $F$. Given a wealth distribution $F$, one first determines the random variable $X_T^{j^*}$ via \eqref{eq: sharpe random variable}. Points along the marginal utility curve are then determined by the first order conditions of the investor's utility maximization problem, which are
\begin{equation}\label{eq: sharpe marginal utility}
 U_T'(X_T^{j^*}(\omega_i))=k \xi_T(\omega_i),\qquad   i=1,\dots,N,
\end{equation}
and
\begin{equation*}
 k\left( C(j^*)-x_0\right)=0,
\end{equation*}
where $k\geq0$ is the Lagrange multiplier associated with the budget constraint $C(j)\leq x_0$. We recall that the strict positivity of the marginal utility function $U_T'$ guarantees that $k>0$ and, therefore, the budget constraint is binding, i.e. $C(j^*)=x_0$. Hence, it is optimal for an expected utility maximizer to select a distribution $F$ whose distributional price in \eqref{eq: sharpe distributional price} is equal to her entire initial budget $x_0$.

To summarize, the investor chooses a distribution $F$ for her terminal wealth that she would like to achieve by investing her initial wealth $x_0>0$. It is assumed that the investor would like to achieve this distribution in a cost-efficient manner and that she implicitly maximizes the expected utility of her terminal wealth. These assumptions then determine the budget constraint that $F$ must satisfy, namely
\begin{equation*}
 x_0 = C(j^*) = \frac{1}{N}\sum_{i=1}^N \xi_T(\omega_i)x_i^F,
\end{equation*}
where $\{x_i^F\}_{i=1}^N$ is the representation of $F$ as an $N$-vector as described above. Furthermore, the pointwise specification of the investor's optimal terminal wealth random variable is given by \eqref{eq: sharpe random variable}. The investor's risk preferences are then described by an $N$-point approximation of the investor's marginal utility curve given by \eqref{eq: sharpe marginal utility}. Finally, the model (a one period model with $N$ possible states) is incomplete for $N>2$, and so it is not possible to uniquely determine the optimal initial allocation $\pi$ to the risky stock.

\subsection{The Distribution Builder interface: How a user selects a desired distribution for her future wealth}

We briefly discuss an example using \emph{The Distribution Builder} so that the reader will be acquainted with one possible procedure for choosing a desired distribution for future wealth. We note, however, that in our continuous-time work herein we assume that the investor chooses a distribution for future wealth, but we do not investigate specific ways or tools she might use for this purpose.

The following example comes from a specific application of \emph{The Distribution Builder}, namely to elicit a desired probability distribution for the user's income per year following retirement.  The interface for this application of \emph{The Distribution Builder} is pictured in figure \ref{fig: distribution builder}. The vertical axis of percentages corresponds to the percentage of pre-retirement income that will be realized annually in retirement. For example, if the investor earned \$100,000 in the year before retirement, the 75\% row corresponds to a subsequent annual retirement income of \$75,000.  

In an experimental setting, users are told that some reference point, which is 75\% in figure \ref{fig: distribution builder}, is a typically recommended goal for annual retirement income. The reference row can then be calibrated to represent the level of wealth that can be attained with certainty by investing in the risk-free asset.

\begin{figure}
\begin{center}
\includegraphics[scale=5.00]{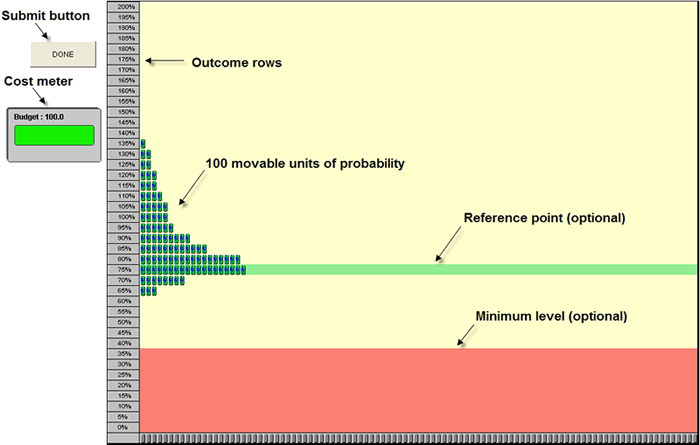}
\end{center}
\caption{The Distribution Builder User Interface. Reprinted from \citet{sha-gol-bly-2000}.}
\label{fig: distribution builder}
\end{figure}

The main area of the interface contains 100 markers, which are initially positioned along the bottom of the screen. Each marker represents an equally-likely state of the world, and the user is told that her realized outcome is represented by one of these markers. \comment{ although she certainly does not know which one in advance.} Users are only able to submit distributions of a given fixed cost (expressed as a percentage), and the cost meter on the left hand side of the interface adjusts accordingly as the user places markers along the vertical axis. The user can submit a distribution of markers only when the cost meter indicates that between 99 and 100 percent of the total fixed budget has been consumed. When satisfied with a particular distribution that meets the cost requirement, the user submits it and the computer then removes all but one of the markers, so that the user is able to experience the actual realization of her desired distribution.

\section{The continuous-time model and background results on investment performance criteria}\label{sect: model and background}

We describe the market setting in which our investor operates, as well as known results on related investment performance criteria. The background results concerning these criteria will be used in the fixed horizon setting in section 4 and the flexible horizon setting in section 5.

The market is complete and consists of a riskless money market and $d$ risky assets driven by $d$ independent Brownian motions.  The risky assets are modeled by time-dependent geometric Brownian motions on $\bbR^d$, i.e. for $i=1,\dots, d$, the price $S^i_t,t\geq0$, of the $i$-th risky asset satisfies
\begin{equation}\label{eq: stock price dynamics}
 dS_t^i = S_t^i \left(\mu^i(t)dt+\sum_{j=1}^d \sigma^{ji}(t)dW_t^j\right),\qquad S_0^i>0,
\end{equation}
where $\mu^i(t)$ and $\sigma^{ji}(t)$ are deterministic functions of time for $i,j=1,\dots, d$, and $t\geq0$. Here, $W=(W^1,\dots,W^d)$ is a $d$-dimensional standard Brownian motion (regarded as a column vector) defined on a complete filtered probability space $(\Omega,\m{F},(\m{F}_t)_{t\geq0},\bbP)$ where the filtration $(\m{F}_t)_{t\geq0}$ satisfies the usual conditions.  It is assumed that $\mu^i(t)$ and $\sigma^{ji}(t)$ are uniformly bounded in $t\geq0$, for all $i,j$.  For brevity, we write $\sigma(t)$ to denote the volatility matrix, i.e. the $d\times d$ matrix $(\sigma^{ji}(t))$ whose $i$-th column represents the volatility $\sigma^i(t)=(\sigma^{1i}(t),\dots,\sigma^{di}(t))$ of the $i$-th risky asset. We also assume that the matrix function $\sigma(t)$ is invertible for all $t\geq0$, and we will write this inverse as $\sigma^{(-1)}(t)$. We can then alternatively write \eqref{eq: stock price dynamics} as
\begin{equation}\label{eq: stock price dynamics 2}
 dS_t^i = S_t^i(\mu^i(t)\;dt+\sigma^i(t)\cdot dW_t).
\end{equation}
The riskless money market has price process $B_t,t\geq0,$ satisfying $B_0=1$ and
\begin{equation}\label{eq: money market price dynamics}
 dB_t = r(t)B_tdt,
\end{equation}
for a nonnegative time-dependent interest rate function $r(t), t\geq0$, which is assumed to be uniformly bounded in $t\geq0$.  We denote by $\mu(t)$ the $d\times 1$ vector with coordinates $\mu^i(t)$ and by $\mathbf{1}$ the $d$-dimensional vector with every component equal to one.

We define the function $\lambda(t),t\geq0,$ by
\begin{equation}\label{eq: lambda}
 \lambda(t):=(\sigma^\top(t))^{(-1)}(\mu(t)-r(t)\mathbf{1}),
\end{equation}
and we will occasionally refer to it as the \emph{market price of risk}.  

\begin{assumption}\label{assump: market price of risk}
The function $\lambda(t),t\geq0,$ is continuous and uniformly bounded on $t\geq0$. Furthermore, its Euclidean norm, $|\lambda(t)|,t\geq0$, is H\"{o}lder continuous, and there exist positive constants $c_0$ and $c_1$ such that $0<c_0\leq |\lambda(t)|\leq c_1$ for all $t\geq0$.
\end{assumption}

Starting at time $t_0=0$ with initial endowment $x_0>0$, the investor invests dynamically in the risky assets and the riskless one. The present values of the amounts invested in the assets are denoted by $\pi_t^i,\;i=1,\dots,d$, and by $\pi_t^0$, respectively. The present value of her total investment is then given by $X_t^\pi = \sum_{i=0}^d \pi_t^i$, which we will refer to as the discounted wealth generated by the (discounted) strategy $\pi=(\pi_t^0,\pi_t^1,\dots,\pi_t^d)$. The investment strategies $\pi$ play the role of control processes and are assumed to be self-financing. Using \eqref{eq: stock price dynamics 2}, \eqref{eq: money market price dynamics} and \eqref{eq: lambda} we deduce
\begin{equation}\label{eq: wealth eqn}
 dX_t^\pi = \sigma(t)\pi_t\cdot ( \lambda(t)dt+dW_t),\qquad t>0,
\end{equation}
where $\pi_t=(\pi_t^i; i=1,\dots, d)$ is a column vector. 

The investor selects a portfolio process from an admissibility set $\m{A}$. A detailed description of this set is given in the upcoming sections.

Finally, we introduce the auxiliary market input processes $A_t$ and $M_t$, $t\geq0$, defined by
\begin{equation}\label{eq: mkt inp proc}
 A_t=\int_0^t \abs{\lambda(s)}^2 ds \qquad\text{and}\qquad M_t=\int_0^t\lambda(s)\cdot dW_s.
\end{equation}
We also recall the martingale $Z_t,t\geq0$, given by
\begin{equation}\label{eq: state price density process}
 Z_t = \exp\left\{ -\int_0^t\lambda(s)\cdot dW_s - \frac{1}{2}\int_0^t\abs{\lambda(s)}^2ds\right\} = \exp\left\{ -M_t-\frac{1}{2}A_t\right\}.
\end{equation}

\subsection{Background results on classical expected utility theory}

We briefly review background results on the classical expected utility theory. These results will be relevant in the fixed horizon setting considered in section 4.

The investor invests in $[0,T]$, with $T>0$ being arbitrary but fixed. She derives utility only from terminal wealth, with objective
\begin{equation}\label{eq: classical value function}
 v(x_0,0):=\sup_{\pi\in\m{A}_T}\bbE\left[U_T(X_T^\pi)|X_0^\pi=x_0\right].
\end{equation}
The set of admissible policies $\m{A}_T$ is defined as the set of $\m{F}_t$-progressively measurable and self-financing portfolio processes $\pi_t,t\in[0,T],$ such that $\bbE\int_0^T\abs{\sigma(s)\pi_s}^2ds<\infty$, and $X_t^\pi\geq0,t\in[0,T],$ $\bbP$-a.s., where $X_t^\pi$ solves \eqref{eq: wealth eqn}. We will call an investor with the above investment paradigm a \emph{Merton investor}. 

The utility function $U_T(\cdot)$ satisfies the following standard assumptions.

\begin{assumption}\label{assumption on utility}
\begin{enumerate}
\item[(i)] The function $U_T\colon(0,\infty)\to\bbR$ is twice continuously differentiable, strictly increasing, and strictly concave.
\item[(ii)]The Inada conditions,
\begin{equation}\label{eq: inada classical utility}
 \lim_{x\downarrow 0}U_T'(x)=\infty\qquad\text{and}\qquad \lim_{x\uparrow\infty}U_T'(x)=0,
\end{equation}
are satisfied
\item[(iii)] The inverse, $I_T\colon(0,\infty)\to(0,\infty)$, of the investor's marginal utility function $U_T'$ has polynomial growth, i.e. there is a constant $\gamma>0$ such that
\begin{equation}\label{eq: I poly growth}
 I_T(y)\leq\gamma + y^{-\gamma}.
\end{equation}
\end{enumerate}
\end{assumption}

The stochastic optimization problem \eqref{eq: classical value function} has been extensively studied and completely solved (see, for example, \citet{kar-shr-1998}). 

The following result relates the Merton investor's optimal wealth process and optimal portfolio process to the solution of the heat equation. It is well known that the optimal policies in this model can be written in terms of a solution to a linear parabolic terminal value problem (see, for example, \citet[Lemma 8.4 (p. 122)]{kar-shr-1998}), but the idea of writing the optimal policies specifically in terms of the solution of the heat equation first appeared in \citet{kal-2011} in the lognormal setting. We state the results of \citet{kal-2011} next.

\begin{proposition}\label{prop: classical utility heat equation}
Let $x_0>0$ be the investor's initial wealth and let $\lambda(t)$ be as in \eqref{eq: lambda}. Let $h\colon\bbR\times[0,T]\to(0,\infty)$ be the unique solution to
\begin{equation}\label{eq: classical utility h pde}
 \left\{ \begin{array}{ll} h_t+\frac{1}{2}|\lambda(t)|^2h_{xx}=0,\quad& (x,t)\in\bbR\times[0,T)\\
          h(x,T)=I_T\left( e^{-x}\right),& x\in\bbR,
         \end{array}\right.
\end{equation}
with $I_T$ satisfying \eqref{eq: I poly growth}. Then, the following hold.

$i)$ The optimal wealth process $X_t^*,t\in[0,T]$, is given by
\begin{equation}\label{eq: classical utility optimal wealth}
 X_t^* = h\left( h^{(-1)}(x_0,0)+A_t+M_t,t\right),\qquad t\in[0,T],
\end{equation}
where $A_t$ and $M_t$, $t\in[0,T]$, are defined in \eqref{eq: mkt inp proc} and $h^{(-1)}$ is the spatial inverse of $h$.  

$ii)$ The optimal portfolio process $\pi^*_t,t\in[0,T],$ that generates $X^*_t$ is given by
\begin{equation}\label{eq: classical utility optimal portfolio}
 \pi_t^* = h_x\left( h^{(-1)}\left(X_t^*,t\right),t\right)\sigma^{(-1)}(t)\lambda(t),\qquad t\in[0,T].
\end{equation}
\end{proposition}

\subsection{Background results on forward investment performance processes}\label{sect: background forward}

We now review results on the so-called forward investment performance process. These results will be relevant for the flexible investment horizon setting of section \ref{sect: flexible}. The forward investment performance process is an investment selection criterion developed by Musiela and Zariphopoulou (see, among others, \citet{mus-zar-2006,mus-zar-2009,mus-zar-2010}) as a complementary alternative to the maximal expected utility theory. The main motivation for this approach is the ability to work in flexible investment horizon settings and define for them time-consistent performance criteria for all times. In this framework, an admissible investment strategy is deemed optimal if it generates a wealth process whose average performance is maintained over time. In other words, the average performance of the optimal strategy at any future date, conditional on today's information, preserves the performance of this strategy up until today. Any strategy that fails to maintain the average performance over 
time is then suboptimal. In contrast to the expected utility criterion considered earlier, the forward investor does not specify her risk preferences for some terminal time. Instead, her risk preferences are specified at initial time by an initial datum $u_0$ and then evolve dynamically forward in time for $t\geq0$. 

Next, we recall the forward investment performance process. The set of admissible strategies, $\m{A}$, is defined to be the set of $\m{F}_t$-progressively measurable and self-financing portfolio processes $\pi_t,t\geq0,$ such that $\bbE\int_0^t\abs{\sigma(s)\pi_s}^2ds<\infty,t>0$, and $X_t^\pi\geq0,t\geq0,\;\bbP-a.s.$, where the discounted wealth process solves \eqref{eq: wealth eqn}.
 
\begin{definition}\label{defn: fpp}
 Let $u_0\colon(0,\infty)\to\bbR$ be strictly concave and strictly increasing. An $\m{F}_t$-adapted process $U(x,t)$ is a forward investment performance if, for $t\geq 0$ and $x\in(0,\infty)$:
\begin{enumerate}
 \item[(i)] $U(x,0)=u_0(x)$,
\item[(ii)] the map $x\mapsto U(x,t)$ is strictly concave and strictly increasing,
\item[(iii)] for each $\pi\in\m{A}$, $\bbE[ U(X_t^\pi,t)^+]<\infty$ and $\bbE[U(X_s^\pi,s)|\m{F}_t]\leq U(X_t^\pi,t),\; s\geq t$,
\item[(iv)] there exists $\pi^*\in\m{A}$ for which $\bbE[U(X_s^{\pi^*},s)|\m{F}_t]=U(X_t^{\pi^*},t),\; s\geq t.$
\end{enumerate}
\end{definition}

We refer the reader to \citet{mus-zar-2009} as well as \citet{kal-2011} for further discussion on the forward investment performance and its similarities and differences with the classical value function.

\subsubsection{Review of monotone forward investment performance processes}\label{sect: forward review monotone}

We focus herein on the class of time-decreasing forward investment performance processes that will be used in our analysis in section \ref{sect: flexible}. These processes were introduced in \citet{mus-zar-2009} and \citet{ber-rog-teh-2009} and further analyzed in \citet{mus-zar-2010}. Therein, it was shown that time-decreasing forward investment performance processes $U(x,t)$ are constructed by compiling market-related input with a deterministic function of space and time. Specifically, for $t\geq 0$, we have
\begin{equation}\label{eq: u fpp}
 U(x,t) = u(x,A_t)
\end{equation}
where $A_t,t\geq 0,$ is as in \eqref{eq: mkt inp proc} and $u(x,t)$ is a smooth function that is spatially strictly increasing and strictly concave, and satisfies
\begin{equation}\label{eq: u pde}
 \left\{ \begin{array}{ll} \displaystyle u_t - \frac{1}{2}\frac{u_x^2}{u_{xx}}=0, \quad& (x,t)\in(0,\infty)\times(0,\infty) \vspace{.5em}\\   u(x,0)=u_0(x), & x\in(0,\infty) \end{array}\right.
\end{equation}
where $u_0\colon(0,\infty)\to\bbR$ is the initial datum of Definition \ref{defn: fpp}.

It is also shown in \citet{mus-zar-2010} that if $h(x,t)$ is defined via the transformation
\begin{equation}\label{eq: u and h relation}
 u_x(h(x,t),t) = e^{-x + \frac{t}{2}},\qquad (x,t)\in\bbR\times[0,\infty),
\end{equation}
then it is a positive and spatially strictly increasing space-time harmonic function, solving the ill-posed heat equation
\begin{equation}\label{eq: h pde}
 \left\{ \begin{array}{ll} \displaystyle h_t + \frac{1}{2}h_{xx}=0, & (x,t)\in \bbR\times[0,\infty)\vspace{.5em}\\ h(x,0)=(u_0')^{(-1)}(e^{-x}),\quad & x\in\bbR.\end{array}\right.
\end{equation}
Moreover, the associated optimal processes $X_t^{*}$ and $\pi_t^{*}$, $t\geq0$, can be written explicitly in terms of market-related input and the function $h$, namely, for $t\geq0$,
\begin{equation}\label{eq: fpp h opt policies}
 X_t^*=h\left(h^{(-1)}(x_0,0)+A_t+M_t,A_t\right)
\end{equation}
and
\begin{equation}\label{eq: fpp h opt policies portfolio}
 \pi_t^*=h_x\left(h^{(-1)}(X_t^*,A_t),A_t\right)\sigma^{(-1)}(t)\lambda(t),
\end{equation}
where $A_t$ and $M_t$, $t\geq0,$ are as in \eqref{eq: mkt inp proc} and the function $h^{(-1)}$ stands for the spatial inverse of $h$.

As mentioned above, problem \eqref{eq: h pde} (and, in turn, \eqref{eq: u pde}) are ill-posed. Nevertheless, as we review next, solutions do exist, though we expect the set of admissible initial data $u(x,0)$ and $h(x,0)$ to be rather restricted. We elaborate on this in Remark \ref{rem: initial conditions u and h}.

From the above, one observes that all objects of interest, including the risk preferences of the investor, her optimal strategies, and the associated forward investment performance process, are determined once the functions $u$ and $h$ are known and the market price of risk is chosen (which yields the processes $A_t$ and $M_t$). The study of the functions $u$ and $h$ is therefore crucial to the understanding of the (forward) portfolio choice problem. 

\begin{remark}
 Recall from Proposition \ref{prop: classical utility heat equation} that a representation of the optimal policies similar to \eqref{eq: fpp h opt policies} and \eqref{eq: fpp h opt policies portfolio} holds in the expected utility case. Note, however, that the harmonic function therein depends on market parameters while, in the monotone forward investment performance case, it does not (cf. \eqref{eq: h nu rep}).
\end{remark}

\subsubsection{Analysis of the functions $u$ and $h$}

We recall some known analytical results concerning the representation of, and connections between, the functions $u$ and $h$. Using Widder's classical theorem, it was shown in \citet{mus-zar-2010} that positive and spatially strictly increasing space-time harmonic functions $h$ can be represented in terms of a Borel measure $\nu$ that has finite Laplace transform and support in the positive reals. Given such a representation, the function $u$ is constructed using \eqref{eq: u and h relation}. Since the risk preferences and optimal strategies of the investor are represented in terms of the functions $u$ and $h$ (cf. \eqref{eq: u fpp}, \eqref{eq: fpp h opt policies}, and \eqref{eq: fpp h opt policies portfolio}), the measure $\nu$ emerges as the defining element in the entire analysis of monotone forward investment performance processes. We specify $\nu$ in detail next.

Define $\m{B}(\bbR^+)$ to be the set of finite Borel measures $\nu$ on $\bbR$ such that $\nu((-\infty,0])=0$, and consider the following subset of $\m{B}(\bbR^+)$:
\begin{eqnarray}
  \m{B}^+(\bbR^+) &=& \left\{ \nu\in\m{B}(\bbR^+)\colon\int_{0+}^{\infty}\frac{\nu(dy)}{y}<\infty \text{ and }\int_{0}^\infty e^{yx}\nu(dy)<\infty,\;  x\in\bbR \right\}.\label{eq: borel meas}
\end{eqnarray}

The following result can be found in \citet{mus-zar-2010}.

\begin{proposition}\label{prop: h nu rep}
 $i)$ Let $\nu\in\m{B}^+(\bbR^+)$. Then, the function $h\colon\bbR\times[0,\infty)\to(0,\infty)$ defined by
\begin{equation}\label{eq: h nu rep}
 h(x,t) = \int_{0+}^\infty \frac{e^{yx-\frac{1}{2}y^2t}}{y}\nu(dy)
\end{equation}
is a solution to \eqref{eq: h pde} that is positive and spatially strictly increasing.

$ii)$ Conversely, let $h\colon\bbR\times[0,\infty)\to(0,\infty)$ be a positive and spatially strictly increasing solution to \eqref{eq: h pde}. Then, there exists $\nu\in\m{B}^+(\bbR^+)$ such that $h$ is given by \eqref{eq: h nu rep}.
\end{proposition}

\begin{remark}
 The proof of Proposition \ref{prop: h nu rep} is based on the classical result of Widder that characterizes nonnegative and spatially strictly increasing solutions to the backward heat equation on the half line $t\in[0,\infty)$ in terms of a Borel measure $\nu$ with finite Laplace transform. An analogous representation result can be obtained in the classical maximal expected utility case for solutions to the related terminal value problem \eqref{eq: classical utility h pde}. Indeed, one can show (see \citet{widder-1975} and \citet{wil-1980}) that $h$ solves \eqref{eq: classical utility h pde} if and only if there exists a Borel measure $\widetilde{\nu}$ on $\bbR$ such that
\begin{equation*}
 \int_{-\infty}^{\infty} e^{-\frac{y^2}{2t}}\;\widetilde{\nu}(dy)<\infty,\qquad t\in(0,T)
\end{equation*}
and
\begin{equation*}
 h(x,t) = \frac{1}{\sqrt{2\pi (A_T-A_t)}}\int_{-\infty}^{\infty} e^{-\frac{1}{2}\frac{(x-y)^2}{(A_T-A_t)}}\;\widetilde{\nu}(dy),\qquad (x,t)\in\bbR\times(0,T).
\end{equation*}
In the expected utility case, we deduce via \eqref{eq: classical utility h pde} that the measure $\widetilde{\nu}$ is absolutely continuous with respect to Lebesgue measure and is given by
\begin{equation*}
 \widetilde{\nu}(dy) = I_T(e^{-y})dy,
\end{equation*}
where $I_T$ is the inverse of the investor's marginal utility $U_T'$. Thus we see from Proposition \ref{prop: classical utility heat equation} that all objects of interest in the classical expected utility model are also specified once the market price of risk and a Borel measure encapsulating the investor's preferences are chosen. A parallel result holds in the monotone forward investment performance case, as we will see below in Theorem \ref{thm: fpp all objects of interest} and Remark \ref{rem: forward nu is preferences}.

\end{remark}

The next result characterizes analytically the set of measures $\m{B}^+(\bbR^+)$ and provides a method by which one can find the measure $\nu$ given the function $h$. It will play a central role in the proof of Theorem \ref{prop: nu gen F}.

\begin{proposition}\label{prop:inv}
i) A Borel measure $\nu$ is in $\m{B}^+(\bbR^+)$ if and only if its Laplace transform is entire and $\int_{0+}^\infty\frac{\nu(dy)}{y}<\infty$. 

ii) Let $h$ be given by \eqref{eq: h nu rep} for some $\nu\in\m{B}^+(\bbR^+)$. The mapping $x\mapsto h_x(x,0)$ is the Laplace transform of $\nu$ and it has a unique analytic extension to $\bbC$. Moreover, the mapping $$x\mapsto h_x(ix,0)$$ is the Fourier transform of $\nu$.  

\end{proposition}

\begin{proof}
 $i)$ If the Laplace transform of $\nu$ is entire, then it is finite for all reals and is therefore in $\m{B}^+(\bbR^+)$. Conversely, if $\nu\in\m{B}^+(\bbR^+)$ then its Laplace transform is finite everywhere and $\nu$ has moments of all orders. The rest of part (i) follows (see, for example, \citet[Lemma 1 in the Appendix]{dyb-rog-1997}). 

$ii)$  Using \eqref{eq: h nu rep}, we differentiate under the integral sign (justified using the dominated convergence theorem) to obtain
\[
 h_x(x,t) = \int_{-\infty}^{\infty} e^{yx-\frac{1}{2}y^2t}\nu(dy).
\]
Thus $x\mapsto h_x(x,0)$ is the Laplace transform of the measure $\nu$. As $\nu\in\m{B}^+(\bbR^+)$, we have by the first part of the Proposition that the Laplace transform is entire. In particular, its extension along the imaginary axis, $x\mapsto h_x(ix,0)$, is the Fourier transform of $\nu$.
\end{proof}

We now recall in detail the one-to-one correspondence between positive and spatially strictly increasing solutions to \eqref{eq: h pde} and spatially strictly increasing and strictly concave solutions to \eqref{eq: u pde}. The following result can be found in \citet{mus-zar-2010}.

\begin{proposition}\label{prop: h u corr}
 $i)$ Let $h$ be a positive and spatially strictly increasing solution to \eqref{eq: h pde} and let $\nu$ be the associated Borel measure (cf. \eqref{eq: h nu rep}). If $\nu$ also satisfies $\nu((0,1])=0$ and $\int_{1+}^{\infty}\frac{\nu(dy)}{y-1}<\infty$, then $u\colon(0,\infty)\times[0,\infty)\to\bbR$ is given by
\begin{equation}\label{eq: u h rep pos 1}
 u(x,t) = -\frac{1}{2}\int_0^t e^{-h^{(-1)}(x,s)+\frac{s}{2}}h_x\left(h^{(-1)}(x,s),s\right)ds + \int_{0}^x e^{-h^{(-1)}(z,0)}dz
\end{equation}
and satisfies
\begin{equation}\label{eq: u zero 1}
 \lim_{x\to 0}u(x,t)=0,\qquad\text{for }t\geq0.
\end{equation}
On the other hand, if $\nu((0,1])>0$ and/or $\int_{1+}^{\infty}\frac{\nu(dy)}{y-1}=\infty$, then
\begin{equation}\label{eq: u h rep pos 2}
 u(x,t) = -\frac{1}{2}\int_0^t e^{-h^{(-1)}(x,s)+\frac{s}{2}}h_x\left(h^{(-1)}(x,s),s\right)ds + \int_{\hat{x}}^x e^{-h^{(-1)}(z,0)}dz,
\end{equation}
for $\hat{x}>0$ with
\begin{equation}\label{eq: u zero 2}
 \lim_{x\to0} u(x,t)=-\infty,\qquad\text{for }t\geq0.
\end{equation}

For each $t\geq0$, the Inada conditions
\begin{equation}\label{eq: u h inada}
 \lim_{x\to 0}u_x(x,t)=\infty\qquad\text{and}\qquad \lim_{x\to\infty}u_x(x,t)=0
\end{equation}
are satisfied for both \eqref{eq: u h rep pos 1} and \eqref{eq: u h rep pos 2}, respectively.

$ii)$ Conversely, let $u\colon(0,\infty)\times[0,\infty)\to\bbR$ be spatially strictly increasing and strictly concave and satisfy \eqref{eq: u pde} as well as the Inada conditions \eqref{eq: u h inada}. If $u$ satisfies \eqref{eq: u zero 1}, then there exists $\nu\in\m{B}^+(\bbR^+)$ satisfying $\nu((0,1])=0$ and $\int_{1+}^{\infty}\frac{\nu(dy)}{y-1}<\infty$ such that $u$ admits representation \eqref{eq: u h rep pos 1} with $h$ given by \eqref{eq: h nu rep}. On the other hand, if $u$ satisfies \eqref{eq: u zero 2}, then there exists $\nu\in\m{B}^+(\bbR^+)$ and either (i) $\nu((0,1])>0$, or (ii) $\nu((0,1])=0$ and $\int_{1+}^{\infty}\frac{\nu(dy)}{y-1}=\infty$, such that $u$ admits representation \eqref{eq: u h rep pos 2} with $h$ given by \eqref{eq: h nu rep}.
\end{proposition}

\begin{remark}\label{rem: initial conditions u and h}
 It follows from Proposition \ref{prop: h u corr} that there exists a monotone forward investment process with initial datum $u_0$ if and only if the initial condition $h(x,0)$ for the space-time harmonic function $h$, associated to $u$ via \eqref{eq: u and h relation}, is given by
\begin{equation*}
 h(x,0) = \int_{0+}^\infty \frac{e^{yx}}{y}\nu(dy),
\end{equation*}
for some $\nu\in\m{B}^+(\bbR^+)$. Therefore, the set of initial conditions for $h$ and, thus of $u$, is restricted to be those functions representable as a particular integral with respect to a Borel measure with finite Laplace transform.
\end{remark}

\subsubsection{Solution to the model under monotone forward investment performance criteria}

We are now ready to recall the characterization of all objects of interest in the case of the monotone forward investment performance criterion. Note that we introduce condition \eqref{eq: nu admissible}, which is a stronger condition than is needed for the representations of $h$ (cf. \eqref{eq: borel meas}) and thus of $u$, but is sufficient to guarantee the admissibility of the candidate optimal policy \eqref{eq: optimal portfolio}. The following result can be found in \citet{mus-zar-2010}.

\begin{theorem}\label{thm: fpp all objects of interest}
 $i)$ Let $h$ be a positive and spatially strictly increasing solution to \eqref{eq: h pde}, for $(x,t)\in\bbR\times[0,\infty)$, and assume that the associated measure $\nu$ satisfies
\begin{equation}\label{eq: nu admissible}
 \int_{-\infty}^{\infty} e^{yx+\frac{1}{2}y^2t}\nu(dy)<\infty,\qquad (x,t)\in\bbR\times[0,\infty).
\end{equation}
Let $A_t$ and $M_t$, $t\geq0$, be as in \eqref{eq: mkt inp proc} and define the processes $X_t^*$ and $\pi_t^*$ by
\begin{equation}\label{eq: optimal wealth}
 X_t^*=h\left(h^{(-1)}(x_0,0)+A_t+M_t,A_t\right) 
\end{equation}
and
\begin{equation}\label{eq: optimal portfolio}
 \pi_t^*=h_x\left(h^{(-1)}(X_t^*,A_t),A_t\right)\sigma^{(-1)}(t)\lambda(t),
\end{equation}
for $t\geq0$, $x_0>0$, with $h$ as above and $h^{(-1)}$ being its spatial inverse. Then, the portfolio process $\pi_t^*$ is admissible and generates $X_t^*$, i.e.
\begin{equation*}
 X_t^* = x_0 + \int_0^t \sigma(s)\pi_s^*\cdot(\lambda(s)ds+dW_s).
\end{equation*}

$ii)$ Let $u$ be a spatially strictly increasing and strictly concave solution to \eqref{eq: u pde}, associated to $h$ via Proposition \ref{prop: h u corr}. Let $U(x,t),t\geq0,x>0$ be given by
\begin{equation}\label{eq: forward performance}
 U(x,t) = u(x,A_t).
\end{equation}
Then $U(x,t)$ is a forward investment performance process and the processes $X_t^*$ and $\pi_t^*$ defined in \eqref{eq: optimal wealth} and \eqref{eq: optimal portfolio} are optimal.
\end{theorem}

\begin{remark}\label{rem: forward nu is preferences}
 The measure $\nu$ encapsulates the investor's risk preferences under monotone forward investment performance criteria. To see this, recall that in the expected utility framework, the investor's initial wealth, market input, and her terminal utility function comprise the set of inputs that are sufficient to solve the investment problem (see Proposition \ref{prop: classical utility heat equation}). On the other hand, under monotone forward investment criteria the sufficient set of inputs is composed of the investor's initial wealth, market input and an admissible Borel measure $\nu$ (rather than a utility function). Indeed, given an admissible measure $\nu$, one forms the function $h$ via \eqref{eq: h nu rep} and the function $u$ via Proposition \ref{prop: h u corr} ($\nu$ also determines the initial datum $u_0$; see Remark \ref{rem: initial conditions u and h}). In turn, one forms the investor's optimal policy and forward investment performance process using Theorem \ref{thm: fpp all objects of interest}.
\end{remark}

To close this section, we present the following scaling result, which shows that one can normalize the function $h$ and assume that the measure $\nu$ is a finite Borel measure of arbitrary total mass. This fact will be used in the proof of Theorem \ref{prop: nu gen F}. To this end, we denote by $h_0$ the total mass of $\nu$ and, with a slight abuse of notation, the associated wealth process by $X_t^*(x_0;h_0),t\geq 0$.

\begin{proposition}\label{prop: scaling}
 For $h_0=\nu(\bbR)$, the optimal wealth process satisfies, for $t\geq 0$,
\begin{equation*}
 \frac{k_0}{h_0}X_t^*(x_0; h_0) = X_t^*\left(\frac{k_0}{h_0}x_0; k_0\right),
\end{equation*}
where $k_0$ is an arbitrary positive constant.
\end{proposition}

\begin{proof}
 Let $\hat{h}(x,t) = \frac{k_0}{h_0}h(x,t)$.  Then,
\begin{eqnarray*}
 X_t^*(x_0;h_0) &=& h\left( h^{(-1)}(x_0,0)+A_t+M_t,A_t\right)=\frac{h_0}{k_0} \hat{h}\left( h^{(-1)}(x_0,0)+A_t+M_t,A_t\right)\\
&=& \frac{h_0}{k_0} \hat{h}\left( \hat{h}^{(-1)}\left( \frac{k_0}{h_0}x_0,0\right)+A_t+M_t,A_t\right)=\frac{h_0}{k_0} X_t^*\left(\frac{k_0}{h_0}x_0;k_0\right),
\end{eqnarray*}
where we have used the fact that $h^{(-1)}(x_0,0) = \hat{h}^{(-1)}\left(\frac{k_0}{h_0}x_0,0\right)$.
\end{proof}

\section{Targeted wealth distributions in a fixed investment horizon setting}\label{sect: fixed}

In this section we consider a Merton investor with the \emph{fixed} investment horizon $[0,T]$, for some arbitrary positive terminal time $T<\infty$. The investment horizon is preset at initial time, when investment begins, and does not change throughout the course of investing. First, we present the case where the investor chooses a probability distribution for her terminal wealth. Subsequently, we consider an investor who chooses a probability distribution for her wealth to be realized at some arbitrary intermediate time within her investment horizon. In both cases, we show how, for a given initial wealth $x_0>0$, the investor's targeted distribution and an estimate of the market price of risk can be used to: 
\begin{itemize}
 \item determine if the chosen distribution is attainable in this market environment;
 \item infer the investor's risk preferences; and
 \item describe how the investor should invest to attain her goal.
\end{itemize}

We start with the family of distributions that we consider herein. Throughout, the function $\Phi\colon\bbR\to(0,1)$ denotes the distribution function of the standard normal random variable.

\begin{assumption}\label{assump: distributions}
 A chosen distribution function $F\colon(0,\infty)\to(0,1)$ for future wealth is continuous, strictly increasing, and satisfies
\begin{equation}\label{eq: distribution growth condition}
 F^{(-1)}(\Phi(x))\leq Ke^{a\abs{x}},\qquad x\in\bbR,
\end{equation}
for some positive constants $K$ and $a$.
\end{assumption}

\subsection{Investment target placed at terminal time}

We start with the case in which the investor specifies a desired distribution for her terminal wealth. We address the three bullet points above. With regards to the second point, we infer the investor's risk preferences by finding her marginal utility function.

\begin{theorem}\label{thm: classical utility infer pref}
 Suppose the investor with initial wealth $x_0>0$ targets her terminal wealth $X_T^*$ to have distribution function $F$ satisfying Assumption \ref{assump: distributions}. Let $A_t$ and $M_t$, $t\in[0,T]$, be as in \eqref{eq: mkt inp proc}. Then, the following hold.

$i)$ The investor's target can be attained only if $F$ satisfies the budget constraint
\begin{equation}\label{eq: F budget constraint}
 x_0=\frac{1}{\sqrt{2\pi A_T}}\int_{-\infty}^{\infty} e^{-\frac{y^2}{2A_T}}F^{(-1)}\left( \Phi\left( \frac{y-A_T}{\sqrt{A_T}}\right)\right)dy,
\end{equation}
where $F^{(-1)}$ denotes the inverse of $F$.

$ii)$ If $F$ satisfies \eqref{eq: F budget constraint}, then the investor's marginal utility function is given by
\begin{equation}\label{eq: classical marginal utility}
 U_T'(x) = \exp\left( -\sqrt{A_T}\Phi^{(-1)}(F(x)) \right).
\end{equation}

$iii)$ The investor's optimal wealth and portfolio processes are given, respectively, by
\begin{equation}\label{eq: thm 1 wealth}
 X_t^* = h\left( h^{(-1)}(x_0,0)+A_t+M_t,t\right),\qquad t\in[0,T],
\end{equation}
and 
\begin{equation}\label{eq: thm 1 portfolio}
 \pi_t^* =  h_x\left( h^{(-1)}\left(X_t^*,t\right),t\right)\sigma^{(-1)}(t)\lambda(t),\qquad t\in[0,T],
\end{equation}
where the function $h$ is given by
\begin{equation}\label{eq: h function fixed horizon}
 h(x,t)=\frac{1}{\sqrt{2\pi(A_T-A_t)}}\int_{-\infty}^{\infty} e^{-\frac{1}{2}\frac{(x-y)^2}{(A_T-A_t)}}F^{(-1)}\left(\Phi\left( \frac{y}{\sqrt{A_T}}\right)\right)dy.
\end{equation}
\end{theorem}

\begin{proof}
If $F$ is the desired wealth distribution function, then \eqref{eq: classical utility optimal wealth} yields
\begin{eqnarray}
 F(y) &=& \bbP(X_T^*\leq y) = \bbP\left(h(h^{(-1)}(x_0,0)+M_T+A_T,T)\leq y\right)\nonumber\\
&=&\bbP\left( M_T\leq h^{(-1)}(y,T)-h^{(-1)}(x_0,0)-A_T\right)\nonumber\\
&=& \label{eq: F in terms of h}\Phi\left( \frac{h^{(-1)}(y,T)-h^{(-1)}(x_0,0)-A_T}{\sqrt{A_T}}\right),
\end{eqnarray}
where we used that $M_T$ is centered normal with variance $A_T$ (see \eqref{eq: mkt inp proc}). 

Next, we choose
\begin{equation}\label{eq: classical utility h inverse}
 h^{(-1)}(x_0,0)=-A_T,
\end{equation}
which, as we explain in detail in Remark \ref{rem: affine transformation}, can be done without loss of generality. From the above and \eqref{eq: F in terms of h}, we then find that
\begin{equation}\label{eq: classical utility terminal pf 1}
 h(x,T) = F^{(-1)}\left(\Phi\left(\frac{x}{\sqrt{A_T}}\right)\right),\qquad x\in\bbR.
\end{equation}

To show i), observe that from \eqref{eq: classical utility optimal wealth}, \eqref{eq: classical utility h inverse}, and \eqref{eq: classical utility terminal pf 1} we have
\begin{equation*}
 X_{T}^* = F^{(-1)}\left(\Phi\left(\frac{M_T}{\sqrt{A_T}}\right)\right).
\end{equation*}
On the other hand, it is well known (see, for example, \citet{kar-shr-1998}) that the budget constraint $x_0=\bbE(Z_TX_T^*)$, where $Z_T$ is as in \eqref{eq: state price density process}, is binding. Combining the above, we deduce that
\begin{eqnarray*}
 x_0&=& \frac{1}{\sqrt{2\pi A_T}}\int_{-\infty}^{\infty} e^{-\frac{y^2}{2A_T}}F^{(-1)}\left(\Phi\left(\frac{y-A_T}{\sqrt{A_T}}\right)\right)dy.
\end{eqnarray*}
Recall that $F$ satisfies the inequality \eqref{eq: distribution growth condition} and, therefore, the above integral converges.

To prove ii), we use equality \eqref{eq: classical utility terminal pf 1} and the terminal condition for $h$ from \eqref{eq: classical utility h pde} to obtain
\begin{equation}\label{eq: classical inverse marginal utility distribution}
 I_T(e^{-x}) = F^{(-1)}\left(\Phi\left( \frac{x}{\sqrt{A_T}}\right)\right),\qquad x\in\bbR.
\end{equation}
Since $I_T=(U_T')^{(-1)}$, we have that
\begin{equation}\label{eq: marginal utility}
 U_T'(x) = \exp\left( - \sqrt{A_T}\Phi^{(-1)}(F(x))\right),
\end{equation}
and \eqref{eq: classical marginal utility} follows.  

We note that the conditions $\lim_{x\downarrow 0}F(x)=0$ and $\lim_{x\uparrow\infty}F(x)=1$ on $F$ ensure that $U_T$ satisfies the Inada conditions \eqref{eq: inada classical utility}.  Moreover, the polynomial growth requirement \eqref{eq: I poly growth} on $I$ necessitates the condition 
\begin{equation}\label{eq: F inverse growth}
 F^{(-1)}(\Phi(x))\leq a +  ae^{b\abs{x}},\qquad x\in\bbR,
\end{equation}
for some positive constants $a$ and $b$, for which \eqref{eq: distribution growth condition} is sufficient.

Finally, to show iii), we recall that the function $h$ satisfies \eqref{eq: classical utility h pde}. Replacing the terminal condition with \eqref{eq: classical inverse marginal utility distribution} and using the representation formula for the solution of the Cauchy problem, we obtain \eqref{eq: h function fixed horizon}.

\end{proof}

\begin{remark}\label{rem: affine transformation}
It is well known that an expected utility maximizer's optimal wealth process is invariant under positively-sloped linear transformations of the utility function $U_T$. This fact leads to a crucial observation used in the proof of Theorem \ref{thm: classical utility infer pref}, namely that the constant $h^{(-1)}(x_0,0)$ can be chosen, without loss of generality, to be any real number. To see this, suppose that the investor has utility function $U_T$. Let $I_T=(U_T')^{(-1)}$ and solve \eqref{eq: classical utility h pde} to obtain $h$, and suppose that $h^{(-1)}(x_0,0)=c_1\in\bbR$. Now let $\widetilde{U}_T(x)=e^{c_1-c_2}U_T(x)$, for some $c_2\in \bbR, c_2\neq c_1$, be a positively-sloped linear transformation of $U_T$. Next, let $\widetilde{I}_T(y) = ( \widetilde{U}_T')^{(-1)}$ and let $\widetilde{h}$ be the solution to \eqref{eq: classical utility h pde} using $\widetilde{I}_T$ in the terminal condition. It is then easily seen that $\widetilde{I}_T(y) = I_T\left(e^{c_2-c_1}y\right)$ and, in turn, that $\widetilde{h}(x,t)= h(x+c_1-c_2,t)$. From this, one observes that the investor's optimal wealth process is invariant under this transformation, that is, using \eqref{eq: classical utility optimal wealth}, we have
\[
 X_t^* = h\left(h^{(-1)}(x_0,0)+A_t+M_t,t\right) = \widetilde{h}\left(\widetilde{h}^{(-1)}(x_0,0)+A_t+M_t,t\right), \quad t\in[0,T],
\]
where $A_t$ and $M_t$, $t\in[0,T]$, are as in \eqref{eq: mkt inp proc}. Moreover, one obtains that $\widetilde{h}^{(-1)}(x_0,0)=c_2$, and we easily conclude.
\end{remark}

\begin{remark}
Recall that in the works of Sharpe \emph{et al.} (see \citet{sha-gol-bly-2000}, \citet{sha-2001}, and \citet{gol-joh-sha-2008}) the market is incomplete. As mentioned in section \ref{sect: sharpe et al}, the developers of \emph{The Distribution Builder} introduce the additional assumption that the investor wants to achieve her distribution in a cost-efficient manner, in that any other investment strategy that achieves the desired distribution costs at least as much. This cost-efficiency property is guaranteed, however, in our complete market setting with an expected utility maximizer over terminal wealth (see \citet{ber-boy-2010}, \citet{dyb-1988a} and \citet{dyb-1988b}). Indeed, a straightforward change of variables shows that the budget constraint \eqref{eq: F budget constraint} can be rewritten as
\begin{equation}\label{eq: classical utility dist price}
 x_0 = \int_0^1 F^{(-1)}_{Z_T}(y)F^{(-1)}(1-y)dy,
\end{equation}
where $F_{Z_T}$ is the distribution function of the state price density $Z_T$ defined in \eqref{eq: state price density process} and $F$ is the investor's desired distribution function as in Theorem \ref{thm: classical utility infer pref}. The significance of this is that the right-hand side of \eqref{eq: classical utility dist price} is known to be the \emph{distributional price} (see \citet{dyb-1988a}), of the distribution $F$ in the given market. That is, among all $\m{F}_T$-measurable random variables $X_T^\pi$ with distribution function $F$ that can be achieved using a strategy $\pi\in\m{A}_T$, the one requiring the least initial endowment is given by the right-hand side of \eqref{eq: classical utility dist price}.  Thus, the investor who maximizes her expected utility also achieves her distributional price.
\end{remark}

\begin{example}\label{ex: term 1}
Suppose the investor aims at acquiring lognormally distributed terminal wealth, i.e. $\log X_T^*$ is centered normal with variance $b$ for some parameter $b>0$. Note that, initially, this choice does not specify a single distribution, but rather a family of distributions parameterized by $b$. The budget constraint \eqref{eq: F budget constraint} then determines the unique $b$ that is consistent with the investor's choice and utility criterion. To this end, it is easily seen that the inequality \eqref{eq: distribution growth condition} is satisfied, and therefore \eqref{eq: F budget constraint} yields that
\begin{equation}\label{eq: ex 1 budget constraint}
x_0 = \frac{1}{\sqrt{2\pi A_T}}\int_{-\infty}^{\infty} \exp\left( -\frac{1}{2A_T}y^2 +\sqrt{\frac{b}{A_T}}y-\sqrt{bA_T}\right)dy =  \exp\left(\frac{b}{2}-\sqrt{bA_T}\right).
\end{equation}
Straightforward manipulation of \eqref{eq: ex 1 budget constraint} yields the following necessary relationship between the investor's wealth and the market, namely
\[
  A_T + 2\log x_0 = b-2\sqrt{b A_T} +A_T = \left(\sqrt{b}-\sqrt{A_T}\right)^2\geq0,
\]
which, in turn, yields that $b = \left(\sqrt{A_T}+\sqrt{A_T+2\log x_0}\right)^2$.

From \eqref{eq: classical marginal utility}, we deduce the investor's marginal utility function, 
\begin{equation*}
U_T'(x)=x^{-\frac{1}{\beta}},\quad\text{with}\quad \beta:=1 + \bigg| 1 - \sqrt{\frac{b}{A_T}}\bigg|.
\end{equation*}
Therefore, we have two cases for the investor's utility function:
\begin{enumerate}
 \item[(a)] If $\beta>1$, then
\begin{equation*}
 U_T(x)=\frac{1}{1-\frac{1}{\beta}} x^{1-\frac{1}{\beta}}.
\end{equation*}
\item[(b)] If $\beta=1$, then $U_T(x)=\log x$.
\end{enumerate}

The underlying harmonic function (see \eqref{eq: h function fixed horizon}) is then given by
\begin{equation*}
 h(x,t)=\exp\left( \beta x+\frac{1}{2}\beta^2(A_T-A_t)\right)
\end{equation*}
and, in turn, \eqref{eq: thm 1 wealth} and \eqref{eq: thm 1 portfolio} yield the optimal policies
\begin{equation*}
 X_t^* = x_0 e^{ \left(\beta-\frac{1}{2}\beta^2\right)A_t+\beta M_t}\qquad\text{and}\qquad \pi_t^* =\beta x_0 e^{\left(\beta-\frac{1}{2}\beta^2\right)A_t+\beta M_t}\sigma^{(-1)}(t)\lambda(t).
\end{equation*}

\end{example}

\begin{example}\label{ex: term 2}
Suppose the investor with initial wealth $x_0>3$ targets that, if $X_T^*$ is her terminal wealth, then the random variable $g(X_T^*)$ has is centered normal with variance $b$ for some $b>A_T$, where
\begin{equation*}
 g(x)=\log(-1+\sqrt{1+x}), \qquad x\in(0,\infty).
\end{equation*}
As in the previous example, this specifies only a family of distributions, and the parameter $b$ is determined through the budget constraint as follows. We have that $F^{(-1)}(\Phi(x))=\exp(2\sqrt{b}x)+2\exp(\sqrt{b}x)$, and so the inequality \eqref{eq: distribution growth condition} is satisfied. The budget constraint \eqref{eq: F budget constraint} then shows the implicit relationship between the parameter $b$ in terms of $x_0$ and $A_T$, namely
\begin{equation}\label{eq: ex 2 budget constraint}
 x_0 = \exp\left( 2(b-\sqrt{bA_T})\right) + 2\exp\left(\frac{b}{2}-\sqrt{bA_T}\right).
\end{equation}
It is easily seen that there is a unique $b$ that satisfies \eqref{eq: ex 2 budget constraint} under our assumptions. From \eqref{eq: classical marginal utility}, the investor's marginal utility function is given by
\begin{equation*}
 U_T'(x) = \left(-1+\sqrt{1+x}\right)^{-\sqrt{\frac{A_T}{b}}}.
\end{equation*}
The underlying harmonic function in \eqref{eq: h function fixed horizon} is
\begin{equation*}
 h(x,t)=\left(\exp\left(2\sqrt{\frac{b}{A_T}}x + 2\frac{b}{A_T}(A_T-A_t)\right)+2\exp\left( \sqrt{\frac{b}{A_T}}x+\frac{1}{2}\frac{b}{A_T}(A_T-A_t)\right)\right).
\end{equation*}
Using the above and \eqref{eq: thm 1 wealth} and \eqref{eq: thm 1 portfolio}, one can find the optimal wealth and portfolio processes.
\end{example}

\subsection{Investment target placed at an intermediate investment time}

In Theorem \ref{thm: classical utility infer pref}, we showed that a Merton investor who specifies her desired distribution for wealth at terminal time $T$ will effectively determine her risk preferences at terminal time, and, in turn, the optimal policy throughout. Next, we consider an investor who specifies a distribution for her wealth to be realized at some arbitrary, but fixed, intermediate time $\widehat{T}\in(0,T)$.

As in Theorem \ref{thm: classical utility infer pref}, we find that the specification of this single distribution at time $\widehat{T}$, when combined with the investor's initial wealth and market input, is sufficient to determine the feasibility of the desired distribution, the optimal policies that achieve the investor's goal, and the investor's risk preferences. The proof relies on the results of Widder on the inversion of the Weierstrass transform (see \citet{hir-wid-1955}).

Before we proceed, we introduce some additional technical assumptions on the investor's chosen distribution.

\begin{assumption}\label{assump: classical intermediate}
 Let $F\colon(0,\infty)\to(0,1)$ be a chosen wealth distribution function. Let $\widehat{T}\in(0,T)$ and recall the function $A_t,t\in[0,T]$, in \eqref{eq: mkt inp proc}. Define the function $G\colon\bbR\to(0,\infty)$ associated to $F$ by
\begin{equation}\label{eq: classical intermediate f}
 G(x):=F^{(-1)}\left(\Phi\left(\frac{cx}{\sqrt{A_{\widehat{T}}}}\right)\right),\quad\text{with} \quad c:=\sqrt{ \frac{A_T-A_{\widehat{T}}}{2}}.
\end{equation}
We assume that:
\begin{enumerate}
 \item[(i)] $G$ extends analytically to an entire function on $\bbC$;
 \item[(ii)] $G$ satisfies the growth condition
\begin{equation*}
 \limsup_{|y|\to\infty} \frac{|G(x+iy)|}{|y|e^{y^2/4}}=0,\qquad\text{uniformly on closed subintervals of }\bbR \text{ containing }x;
\end{equation*}
 \item[(iii)] The function $g\colon\bbR\times(0,1)\to\bbC$ defined by  
\begin{equation}\label{eq: g intermediate}
 g(x,t):= \frac{1}{\sqrt{4\pi t}}\int_{-\infty}^\infty e^{-y^2/4t}G(x+iy)dy
\end{equation}
is real-valued and nonnegative for all $(x,t)\in\bbR\times(0,1)$.
\end{enumerate}
\end{assumption}

We are now ready to state the results. We recall that $I_T\colon(0,\infty)\to(0,\infty)$ is the inverse of the investor's marginal utility function $U_T'\colon(0,\infty)\to(0,\infty)$, and that $\Phi\colon\bbR\to(0,1)$ denotes the distribution function of the standard normal random variable.

\begin{theorem}\label{thm: classical utility intermediate time}
Suppose the investor with initial wealth $x_0>0$ targets her wealth $X_{\widehat{T}}^*$, at some intermediate time $\widehat{T}\in(0,T)$, to have distribution function $F$ satisfying Assumption \ref{assump: distributions}. Let $A_t$ and $M_t$, $t\in[0,T]$, be as in \eqref{eq: mkt inp proc}. Then, the following hold.

$i)$ The investor's target can be attained only if $F$ satisfies the budget constraint
\begin{equation}\label{eq: F budget constraint intermediate}
 x_0=\frac{1}{\sqrt{2\pi A_{\widehat{T}}}}\int_{-\infty}^{\infty} e^{-\frac{y^2}{2A_{\widehat{T}}}}F^{(-1)}\left( \Phi\left( \frac{y-A_{\widehat{T}}}{\sqrt{A_{\widehat{T}}}}\right)\right)dy.
\end{equation}

$ii)$ If $F$ satisfies \eqref{eq: F budget constraint intermediate} and, in addition, Assumption \ref{assump: classical intermediate}, then the inverse $I_T$ of the investor's marginal utility function is given by
\begin{equation}\label{eq: classical intermediate result I}
 I_T(x)=\frac{1}{\sqrt{2\pi(A_T-A_{\widehat{T}})}}\int_{-\infty}^\infty e^{-\frac{1}{2}\frac{y^2}{(A_T-A_{\widehat{T}})}}F^{(-1)}\left(\Phi\left(\frac{-\log x+iy}{\sqrt{A_{\widehat{T}}}}\right)\right)dy.
\end{equation}

$iii)$ If $F$ satisfies \eqref{eq: F budget constraint intermediate} and, in addition, Assumption \ref{assump: classical intermediate}, then the investor's optimal wealth and portfolio processes are given by
\begin{equation}\label{eq: thm 2 wealth}
 X_t^* =  h\left( h^{(-1)}(x_0,0)+A_t+M_t,t\right), \qquad t\in[0,T],
\end{equation}
and 
\begin{equation}\label{eq: thm 2 portfolio}
 \pi_t^* =  h_x\left( h^{(-1)}\left(X_t^*,t\right),t\right)\sigma^{(-1)}(t)\lambda(t), \qquad t\in[0,T],
\end{equation}
respectively, where the function $h$ is given by
\begin{equation}\label{eq: thm 2 h}
 h(x,t) = \frac{1}{\sqrt{2\pi(A_T-A_t)}}\int_{-\infty}^\infty e^{-\frac{1}{2}\frac{(x-y)^2}{(A_T-A_t)}}I_T(e^{-y})dy,
\end{equation}
with $I_T$ as in \eqref{eq: classical intermediate result I}.

\end{theorem}

\begin{proof}
Recall that although the investor is specifying desired distributional data at time $\widehat{T}\in(0,T)$, her investment horizon is $[0,T]$.  If the investor targets her wealth at time $\widehat{T}$ to have distribution function $F$, then \eqref{eq: classical utility optimal wealth} yields
\begin{eqnarray*}
 F(y) &=& \bbP(X_{\widehat{T}}^*\leq y) = \bbP\left(h\left(h^{(-1)}(x_0,0)+M_{\widehat{T}}+A_{\widehat{T}},\widehat{T}\right)\leq y\right)\\
&=& \Phi\left( \frac{h^{(-1)}(y,\widehat{T})-h^{(-1)}(x_0,0)-A_{\widehat{T}}}{\sqrt{A_{\widehat{T}}}}\right),
\end{eqnarray*}
where we used that $M_{\widehat{T}}$ is centered normal with variance $A_{\widehat{T}}$. Inverting, we deduce that
\begin{equation}\label{eq: classical utility intermediate pf 3}
 h(x,\widehat{T}) = F^{(-1)}\left(\Phi\left(\frac{x}{\sqrt{A_{\widehat{T}}}}\right)\right),\qquad x\in\bbR,
\end{equation}
where, in analogy to the proof of Theorem \ref{thm: classical utility infer pref} and Remark \ref{rem: affine transformation}, we have chosen 
\begin{equation}\label{eq: classical utility intermediate pf h inverse}
h^{(-1)}(x_0,0)=-A_{\widehat{T}}.
\end{equation}

To show i), observe that from \eqref{eq: classical utility optimal wealth}, \eqref{eq: classical utility intermediate pf 3}, and \eqref{eq: classical utility intermediate pf h inverse} we have
\begin{equation}\label{eq: pf intermediate} 
 X_{\widehat{T}}^* = F^{(-1)}\left(\Phi\left(\frac{M_{\widehat{T}}}{\sqrt{A_{\widehat{T}}}}\right)\right).
\end{equation}
Recall $Z_{\widehat{T}}$ from \eqref{eq: state price density process}. Then, \eqref{eq: pf intermediate} yields
\begin{eqnarray*}
 x_0&=&\bbE(Z_{\widehat{T}}X_{\widehat{T}}^*)= \frac{1}{\sqrt{2\pi A_{\widehat{T}}}}\int_{-\infty}^{\infty} e^{-\frac{y^2}{2A_{\widehat{T}}}}F^{(-1)}\left(\Phi\left(\frac{y-A_{\widehat{T}}}{\sqrt{A_{\widehat{T}}}}\right)\right)dy,
\end{eqnarray*}
where the first equality is due to the well-known budget constraint in this model (see, for example, \citet{kar-shr-1998}) and the fact that $Z_tX_t^*, t\in[0,T],$ is a $\bbP$-martingale. Recall that $F$ satisfies the growth condition \eqref{eq: distribution growth condition}, and thus the above integral converges.

To prove ii), first note that by \eqref{eq: classical utility intermediate pf 3} and the uniqueness of the solution to \eqref{eq: classical utility h pde}, we must have
\begin{equation}\label{eq: classical utility intermediate pf 2}
F^{(-1)}\left(\Phi\left(\frac{x}{\sqrt{A_{\widehat{T}}}}\right)\right)=\frac{1}{\sqrt{2\pi(A_T-A_{\widehat{T}})}}\int_{-\infty}^{\infty} e^{-\frac{1}{2}\frac{(x-y)^2}{(A_T-A_{\widehat{T}})}}I_T(e^{-y})dy,\qquad x\in\bbR.
\end{equation}
By a change of variabless, we deduce that this is equivalent to 
\begin{equation}\label{eq: classical utility intermediate pf 4}
 F^{(-1)}\left(\Phi\left(\frac{cx}{\sqrt{A_{\widehat{T}}}}\right)\right)=\frac{1}{\sqrt{4\pi}}\int_{-\infty}^{\infty} e^{-(x-y)^2/4}I_T\left(e^{-cy}\right)dy,
\end{equation}
where $c:=\left(\frac{1}{2}(A_T-A_{\widehat{T}})\right)^{\frac{1}{2}}$. Next, we note that the right-hand side of \eqref{eq: classical utility intermediate pf 4} is the Weierstrass transform of the function $x\mapsto I(e^{-cx})$.  By \citet[Theorem 12.4 (p. 204); see also Definition 3.2 and Theorem 3.2 (p. 180)]{hir-wid-1955}, for such a representation to exist and to converge for all $x\in\bbR$, it is necessary and sufficient that the function $G\colon\bbR\to(0,\infty)$, defined in \eqref{eq: classical intermediate f}, satisfies conditions (i), (ii) and (iii) of Assumption \ref{assump: classical intermediate}. Under these conditions, Widder's theorem on the inversion of the Weierstrass transform (see \citet[Theorem 7.4 p. 191]{hir-wid-1955}) yields that
\begin{equation}\label{eq: classical utility intermediate pf 1}
 I_T(e^{-cx}) = \lim_{t\uparrow 1}g(x,t)= \lim_{t\uparrow 1}\frac{1}{\sqrt{4\pi t}}\int_{-\infty}^{\infty} e^{-y^2/4t}G(x+iy)dy,\qquad \text{a.e. }x\in\bbR,
\end{equation}
with $g$ as in \eqref{eq: g intermediate}. On the other hand, because both sides of \eqref{eq: classical utility intermediate pf 1} are continuous in $x$, this equality holds for all $x\in\bbR$. Moreover, since $G$ satisfies the growth condition (ii) of Assumption \ref{assump: classical intermediate}, the integral in \eqref{eq: classical utility intermediate pf 1} is dominated by
\begin{equation*}
 \int_{-\infty}^\infty |y|e^{-y^2/4t}e^{-y^2/4}dy=\int_{-\infty}^{\infty} |y| e^{-\frac{(1+t)}{4t}y^2}dy\leq\int_{-\infty}^{\infty}|y|e^{-y^2/2}dy.
\end{equation*}
Since the dominant integral converges and is independent of $t$, we have by the dominated convergence theorem that 
\[
 I_T(e^{-cx}) = \lim_{t\uparrow 1}g(x,t) = \frac{1}{\sqrt{4\pi}}\int_{-\infty}^{\infty} e^{-y^2/4}G(x+iy)dy,
\]
which yields \eqref{eq: classical intermediate result I} after a change of variables.

Finally, part iii) follows from the representation formula for the solution of the Cauchy problem \eqref{eq: classical utility h pde}.
\end{proof}

\begin{example}
Suppose the investor desires lognormally distributed wealth at time $\widehat{T}\in(0,T)$, i.e. $\log X_{\widehat{T}}^*$ is centered normal with variance $b$ for some $b>0$. As in Example \ref{ex: term 1}, we note that this specifies only a family of distributions. The budget constraint \eqref{eq: F budget constraint intermediate} implies that 
\[
 A_{\widehat{T}} + 2\log x_0 = b-2\sqrt{b A_{\widehat{T}}} +A_{\widehat{T}} = \left(\sqrt{b}-\sqrt{A_{\widehat{T}}}\right)^2\geq0,
\]
and therefore, the distribution that is consistent with the investor's choice and criterion has parameter $b$ given uniquely by
\begin{equation*}
b = \left( \sqrt{A_{\widehat{T}}} + \sqrt{A_{\widehat{T}} + 2\log x_0} \right)^2. 
\end{equation*}
The function $G$ (see \eqref{eq: classical intermediate f}) then becomes
\begin{equation*}
 G(x)= e^{kx},\quad\text{with}\quad k:= \sqrt{\frac{A_T-A_{\widehat{T}}}{2}}\left(1 + \bigg| 1 - \sqrt{\frac{b}{A_{\widehat{T}}}}\bigg|\right).
\end{equation*}
This function satisfies (i), (ii) and (iii) of Assumption \ref{assump: classical intermediate}.  

Using \eqref{eq: classical intermediate result I}, we easily see that the inverse of the investor's marginal utility is given by
\begin{equation*}
 I_T(e^{-x})= e^{\beta x-k^2},\quad\text{with}\quad \beta :=1 + \bigg| 1 - \sqrt{\frac{b}{A_{\widehat{T}}}}\bigg|.
\end{equation*}
Therefore, the investor's marginal utility function is given by
\begin{equation*}
 U_T'(x) = e^{-\frac{1}{\beta}}x^{-\frac{1}{\beta}},
\end{equation*}
while the underlying harmonic function (see \eqref{eq: thm 2 h}) is $h(x,t)=e^{-k^2}\exp\left( \beta x + \frac{1}{2}\beta^2(A_T-A_t)\right).$ Hence \eqref{eq: thm 2 wealth} and \eqref{eq: thm 2 portfolio} yield the optimal policies
\begin{equation*}
 X_t^* = x_0 e^{ \left(\beta-\frac{1}{2}\beta^2\right)A_t+\beta M_t}\qquad\text{and}\qquad \pi_t^* =\beta x_0 e^{\left(\beta-\frac{1}{2}\beta^2\right)A_t+\beta M_t}\sigma^{(-1)}(t)\lambda(t).
\end{equation*}

\end{example}

\section{Targeted wealth distributions in a flexible horizon setting}\label{sect: flexible}

We continue our study of how an investor's desired distribution for future wealth can be used to recover her risk preferences and construct her optimal policies. In the previous section, we considered a Merton investor with a fixed investment horizon $[0,T]$. In this section, we allow for the investor to have flexibility in her investment horizon. There are practical reasons for allowing such flexibility. For instance, the investor may not know \emph{a priori} until when she will be investing, or may wish to invest indefinitely, or may wish have the flexibility to roll over her portfolio or otherwise extend her investment horizon beyond the original prespecified terminal time. Flexibility in the investment horizon falls outside the classical fixed-horizon Merton problem. An appropriate investment criterion is instead the forward investment performance framework, which we reviewed in section \ref{sect: background forward}. Similar to the fixed horizon setting of section \ref{sect: fixed}, we show how the investor's targeted distribution, her initial wealth, and an estimate of the market price of risk can be used to: 
\begin{itemize}
 \item determine if the chosen distribution is attainable in this market environment;
 \item infer the investor's risk preferences at initial time and describe how they change dynamically throughout the investment horizon; and
 \item describe how the investor should invest in the market to attain her goal.
\end{itemize}

\subsection{Investment target at an arbitrary point for an investor without a fixed terminal horizon}\label{subsect: flex}

We consider an investor in a flexible investment horizon setting who places a desired distribution for wealth at some fixed, but arbitrary, future time. The following result shows that the investor's desired distribution for future wealth, when combined with her initial wealth and market input, determines the Fourier transform of a Borel measure $\nu\in\m{B}^+(\bbR^+)$, where $\m{B}^+(\bbR^+)$ is as in \eqref{eq: borel meas}. As discussed in section \ref{sect: forward review monotone}, this measure is the defining element for the functions $u$ and $h$ in the monotone forward investment performance framework. If, in addition, the measure satisfies \eqref{eq: nu admissible}, then one can also find the optimal wealth process, the optimal investment strategy $\pi^*$ that achieves it, and the forward investment performance process via \eqref{eq: optimal wealth}, \eqref{eq: optimal portfolio}, and \eqref{eq: forward performance}, respectively.

We recall that the function $\Phi\colon\bbR\to(0,1)$ denotes the distribution function of the standard normal random variable, and we denote by $\phi$ its density.

\begin{theorem}\label{prop: nu gen F}
Suppose the investor with initial wealth $x_0>0$ targets her wealth $X_T^*$ at some time $T\in(0,\infty)$ to have distribution function $F$ satisfying Assumption \ref{assump: distributions}. Let $A_t$ and $M_t$, $t\geq0$, be as in \eqref{eq: mkt inp proc}. Then, the following hold.

$i)$ The investor's target can be attained only if $F$ satisfies the budget constraint
\begin{equation}\label{eq: F budget constraint flexible}
 x_0 = \frac{1}{\sqrt{2\pi A_T}}\int_{-\infty}^{\infty} e^{-\frac{y^2}{2A_T}} F^{(-1)}\left( \Phi\left( \frac{y-A_T}{\sqrt{A_T}}\right)\right)dy.
\end{equation}

$ii)$ If $F$ satisfies \eqref{eq: F budget constraint flexible}, then the Fourier transform of the underlying measure $\nu$ is given by
\begin{equation}\label{eq: ft gen F}
 \varphi_\nu(x) =\frac{1}{A_T\sqrt{2\pi}} \int_{-\infty}^{\infty} e^{- \frac{(ix-y)^2}{2A_T}} \frac{\phi \left( \frac{y}{\sqrt{A_T}} \right) }{ f\left( F\left( \Phi\left( \frac{y}{\sqrt{A_T}} \right)\right)\right)}dy,
\end{equation}
where $f$ is the density of $F$. Moreover, if $u_0$ is the investor's initial datum, then
\begin{equation}\label{eq: initial datum}
 u_0'(x) = \exp\left(-h_0^{(-1)}(x)\right),
\end{equation}
with $h_0$ given by
\begin{equation}\label{eq: h initial datum}
 h_0(x)=\int_{0+}^{\infty} \frac{e^{yx}}{y}\nu(dy).
\end{equation}

$iii)$ If the above measure $\nu$ satisfies \eqref{eq: nu admissible}, then the investor's optimal wealth and portfolio processes are given by
\begin{equation}\label{eq: forward optimal wealth}
 X_t^* = h\left( h^{(-1)}(x_0,0)+A_t+M_t,A_t\right),\qquad t\geq0,
\end{equation}
and 
\begin{equation}\label{eq: forward optimal portfolio}
 \pi_t^* =  h_x\left( h^{(-1)}\left(X_t^*,A_t\right),A_t\right)\sigma^{(-1)}(t)\lambda(t),\qquad t\geq0,
\end{equation}
respectively, where $h$ is given by
\begin{equation}\label{eq: thm 3 h}
 h(x,t) = \int_{0+}^{\infty} \frac{e^{yx-\frac{1}{2}y^2t}}{y}\nu(dy).
\end{equation}

\end{theorem}

\begin{proof}
Let $h(x,t)$ be given by \eqref{eq: h nu rep} for some $\nu\in\m{B}^+(\bbR^+)$. Recall from Proposition \ref{prop: scaling} that we can assume, without loss of generality, that $\nu$ has arbitrary total mass. Therefore, we assume that $\nu$ is such that it satisfies $\int_{0+}^{\infty}\frac{e^{-A_Ty}}{y}\nu(dy)=x_0$ or, equivalently, that
\begin{equation}\label{eq: flexible h inverse}
 h^{(-1)}(x_0,0)=-A_T.
\end{equation}
Then, using \eqref{eq: optimal wealth}, we obtain that
\begin{equation}\label{eq: wealth h pf 1}
 X_T^*  = h(M_T,A_T).
\end{equation}
If the investor targets her wealth at time $T$ to have distribution function $F$, then using that $M_T$ is centered normal with variance $A_T$, we deduce that
\begin{equation}\label{eq: dist ow}
 F(y)=\bbP(X_T^{*}\leq y) = \Phi\left(\frac{ h^{(-1)}\left(y,A_T\right)}{\sqrt{A_T}}\right)
\end{equation}
and, in turn, that
\begin{equation}\label{eq: thm pf 1}
 h(x,A_T)=F^{(-1)}\left( \Phi\left( \frac{x}{\sqrt{A_T}}\right)\right).
\end{equation}
Part i) then follows from the well-known budget constraint in this model (see, for example, \citet{kar-shr-1998}), \eqref{eq: wealth h pf 1}, \eqref{eq: thm pf 1}, \eqref{eq: state price density process}, \eqref{eq: flexible h inverse} and \eqref{eq: distribution growth condition}.

We now prove ii). By \eqref{eq: h pde} and \eqref{eq: thm pf 1}, the function $h$ must solve
\begin{equation}\label{eq: h pde pf flexible}
 \left\{\begin{array}{ll} h_t+\frac{1}{2}h_{xx}=0,& (x,t)\in(0,\infty)\times(0,A_T)\\
         h(x,A_T) = F^{(-1)}\left( \Phi\left( \frac{x}{\sqrt{A_T}}\right)\right),\quad & x\in(0,\infty).
        \end{array}\right.
\end{equation}
Condition \eqref{eq: distribution growth condition} implies that the terminal data satisfies the Tychonov condition (see \citet[Chapter 1]{fri-1964}) and so the unique solution to \eqref{eq: h pde pf flexible} is given by the convolution
\begin{equation*}
 h(x,t) = \frac{1}{\sqrt{2\pi(A_T-t)}}\int_{-\infty}^{\infty} e^{-\frac{1}{2}\frac{(x-y)^2}{(A_T-t)}} h(y,A_T)dy.
\end{equation*}
Since $x\mapsto h(x,A_T)$ is differentiable almost everywhere, we obtain
\begin{equation*}
 h_x(x,t) = \frac{1}{\sqrt{2\pi(A_T-t)}}\int_{-\infty}^{\infty} e^{-\frac{1}{2}\frac{(x-y)^2}{(A_T-t)}} h_x(y,A_T)dy.
\end{equation*}
By Proposition \ref{prop:inv}, we then conclude that the function $\varphi_\nu\colon\bbR\to\bbC$ given by
\begin{eqnarray*}
 \varphi_\nu(x) &=& h_x(ix,0) \\
&=&\frac{1}{\sqrt{2\pi A_T^2}} \int_{-\infty}^{\infty} e^{- \frac{(ix-y)^2}{2A_T}} \frac{\phi \left( \frac{y}{\sqrt{A_T}} \right) }{ f\left( F\left( \Phi\left( \frac{y}{\sqrt{A_T}} \right)\right)\right)}dy 
\end{eqnarray*}
is the Fourier transform of the implied measure $\nu$. Equations \eqref{eq: initial datum} and \eqref{eq: h initial datum} then follow from \eqref{eq: u pde}, \eqref{eq: u and h relation}, and Proposition \ref{prop: h nu rep}.

Part iii) follows by Theorem \ref{thm: fpp all objects of interest} and \eqref{eq: h nu rep}.

\end{proof}

\begin{remark}\label{rem: weaker growth condition F}
 The growth assumption \eqref{eq: distribution growth condition} for the distribution $F$ in Assumption \ref{assump: distributions} can be slightly relaxed. Indeed, in order for the Tychonov condition to be satisfied in \eqref{eq: h pde pf flexible}, it is sufficient that
\begin{equation}\label{eq: F inverse growth weaker}
 F^{(-1)}(\Phi(x))\leq K e^{a x^2},\qquad   x\in\bbR,
\end{equation}
for some positive constants $K$ and $a<\frac{1}{2}$. In Example \ref{ex: 3}, we analyze a case in which $F$ satisfies \eqref{eq: F inverse growth weaker} but not \eqref{eq: distribution growth condition}.
\end{remark}

\begin{example}
Suppose that the investor desires lognormally distributed wealth at time $T$, i.e. $\log X_{T}^{*}$ is centered normal with variance $b$ for some $b>0$. Working as in the previous examples, in order to specify the distribution that is consistent with the investor's choice and criterion, we use the budget constraint \eqref{eq: F budget constraint flexible} to find that
\[
 A_T + 2\log x_0 = b-2\sqrt{b A_T} +A_T = \left(\sqrt{b}-\sqrt{A_T}\right)^2\geq0.
\]
Thus, $b$ is given uniquely by
\begin{equation*}
b = \left( \sqrt{A_T} + \sqrt{A_T + 2\log x_0}\right)^2.
\end{equation*}
Using this and \eqref{eq: ft gen F}, the Fourier transform of $\nu$ is then given by
\begin{equation*}
 \varphi_\nu(x)=\beta\exp\left(ix\beta\right),\quad\text{with}\quad \beta:=1 + \bigg| 1 - \sqrt{\frac{b}{A_T}}\bigg|.
\end{equation*}
We easily see that this is the Fourier transform of the Dirac point mass $\nu = \beta\delta_{\beta}$, which satisfies the admissibility condition \eqref{eq: nu admissible}. Using \eqref{eq: initial datum} and \eqref{eq: h initial datum}, we find that $u_0'(x) = x^{-\frac{1}{\beta}}$ and, using \eqref{eq: h nu rep}, we deduce that $h(x,t)= e^{\beta x-\frac{1}{2}\beta^2 t}$.

The associated optimal wealth and portfolio processes are given by
\begin{equation*}
 X_t^* = x_0e^{\left( \beta-\frac{1}{2}\beta^2\right)A_t + \beta M_t}\qquad \text{and}\qquad \pi_t^* = \beta x_0e^{\left( \beta-\frac{1}{2}\beta^2\right)A_t + \beta M_t}\sigma^{(-1)}(t)\lambda(t).
\end{equation*}
Finally, we deduce the investor's forward investment performance process. If, for instance, $\beta>1$, the investor's forward investment performance is given by
\begin{equation*}
 U(x,t) = \frac{ \beta^{\frac{\beta-1}{\beta}}}{\beta-1} x^{\frac{\beta-1}{\beta}}e^{-\frac{\beta-1}{2}A_t}.
\end{equation*}
\end{example}

\begin{example}
 Suppose that the investor with initial wealth $x_0>3$ desires that, if $X_T^*$ is her wealth at time $T$, then the random variable $g(X_T^*)$ is centered normal with variance $b$ for some $b>A_T$, where $g\colon(0,\infty)\to\bbR$ is given by $g(x) = \log(-1+\sqrt{1+x})$. Again, note that this is a family of distributions. Using the budget constraint \eqref{eq: F budget constraint flexible} we find that
\begin{equation}\label{eq: ex 5 budget constraint}
 x_0 = \exp\left( 2\left( b-\sqrt{A_Tb}\right)\right) + 2\exp\left( \frac{b}{2}-\sqrt{A_Tb}\right).
\end{equation}
Under our assumptions, it is easily seen that there is a unique $b$ that satisfies \eqref{eq: ex 5 budget constraint}.

Next, the Fourier transform of the implied measure $\nu$ is found via \eqref{eq: ft gen F}. Specifically, 
\begin{eqnarray*}
 \varphi_\nu(x)&=& 2\sqrt{\frac{b}{A_T}}\left(\exp\left(2ix\sqrt{\frac{b}{A_T}}+2b-2\sqrt{bA_T}\right)+\exp\left(ix\sqrt{\frac{b}{A_T}}+\frac{b}{2}+\sqrt{bA_T}\right)\right)\\
&=& \sqrt{\frac{b}{A_T}}\left(k_1 \exp\left(2ix\sqrt{\frac{b}{A_T}}\right) + k_2 \exp\left(ix\sqrt{\frac{b}{A_T}}\right)\right),
\end{eqnarray*}
where the constants $k_1$ and $k_2$ are given by
\[
 k_1 = \frac{2e^{2b-2\sqrt{bA_T}}}{e^{2b-2\sqrt{bA_T}}+2e^{b/2-\sqrt{bA_T}}}\qquad\text{and}\qquad k_2 = \frac{2 e^{b/2-\sqrt{bA_T}}}{e^{2b-2\sqrt{bA_T}}+2e^{b/2-\sqrt{bA_T}}}.
\]
The implied measure $\nu$ is then given by the sum of Dirac point masses:
\[
 \nu = k_1\beta\delta_{2\beta} + k_2\beta \delta_{\beta},\qquad\text{with}\qquad \beta=\sqrt{\frac{b}{A_T}}.
\]
Using \eqref{eq: h initial datum} and \eqref{eq: initial datum}, we in turn deduce that
\[
u_0'(x) = \left( \sqrt{ \frac{2}{k_1}x +\frac{k_2^2}{k_1^2} } - \frac{k_2}{k_1} \right)^{-\frac{1}{\beta}}. 
\]
Moreover, it easily follows that $\nu$ satisfies \eqref{eq: nu admissible}. Using \eqref{eq: thm 3 h} we then find
\[
 h(x,t) = \frac{k_1}{2}e^{2\beta x-2\beta^2t}+k_2e^{\beta x-\frac{1}{2}\beta^2t}.
\]
From there, one can apply formulae \eqref{eq: forward optimal wealth}, \eqref{eq: forward optimal portfolio}, and \eqref{eq: forward performance}, to find the optimal wealth process, the optimal investment policy that generates it, and the forward investment performance process that are consistent with the investor's preferences.

\end{example}

We conclude this section by considering one case where the range of the investor's wealth is the entire real line. Although we do not systematically consider investment problems in which wealth can become negative herein, we nevertheless provide an informal example. It was shown in \citet{mus-zar-2010} that representation results for the optimal policies and the forward investment performance process in terms of a Borel measure $\nu$ hold in the case of possibly negative wealth. These representation results are similar to those in subsection \ref{subsect: flex} in the case of nonnegative wealth.

\begin{example}\label{ex: 3}
Suppose the investor targets her wealth at time $T$ to have distribution function
\[
 F(y) = \Phi\left(\sqrt{1+1/A_T}H^{(-1)}(y)\right),
\]
where $\Phi$ is the standard normal distribution function, and $H\colon\bbR\to\bbR$ is given by
\[
 H(x)=\int_0^x e^{\frac{1}{2}z^2}dz.
\]
We assume that the investor's initial wealth is such that the budget constraint \eqref{eq: F budget constraint flexible} is satisfied. Note that $F$ satisfies the growth condition \eqref{eq: F inverse growth weaker} but not \eqref{eq: distribution growth condition}. Nevertheless, as mentioned in Remark \ref{rem: weaker growth condition F}, the conclusions of Theorem \ref{prop: nu gen F} hold.

After some tedious but straightforward calculations, we deduce via \eqref{eq: ft gen F} that the Fourier transform of the implied measure $\nu$ is given by
\[
\varphi_\nu(x)=e^{-\frac{1}{2}x^2}.
\]
This is the characteristic function of a standard normal random variable, and so $\nu(dy)=\frac{1}{\sqrt{2\pi}}e^{-\frac{1}{2}y^2}dy$. Note, however, that this measure $\nu$ violates condition \eqref{eq: nu admissible} for $t>0$ and satisfies only \eqref{eq: borel meas}. In this situation, one can work with the so-called local forward investment performance process (see \citet[section 2]{mus-zar-2010}).

\end{example}

\section{Comments and conclusions}\label{sect: conclusions}

\subsection{Time-consistency of distributional investment targets}

Besides the feasibility conditions we considered in sections \ref{sect: fixed} and \ref{sect: flexible}, it is natural to investigate whether an investor who desires a certain wealth profile at time $T_1$ can also choose a wealth profile at a different time $T_2, T_1\neq T_2$. The market model considered herein, however, is not general enough to allow for this to be done in an arbitrary way. Indeed, Theorems \ref{thm: classical utility infer pref}, \ref{thm: classical utility intermediate time}, and \ref{prop: nu gen F} demonstrate that, along with the investor's initial wealth and market input, the specification of a \emph{single} desired distribution for future wealth fully determines the investor's optimal wealth process at all times within the investor's investment horizon. Hence, the investor in the market considered herein is only permitted to choose a distribution for wealth at one future time, in both the fixed and flexible horizon settings. This choice determines her wealth process pointwise, and thus in distribution, at all other times.

\subsection{Role of initial wealth}

The investor's initial wealth $x_0$ plays an important, albeit subtle, role in our work. The choice of $x_0$ is arbitrary but fixed throughout the paper. The initial wealth, together with the investor's choice of distribution and market input, comprises the set of necessary inputs for the analysis. Indeed, the three inputs are interrelated via the budget constraints (see \eqref{eq: F budget constraint}, \eqref{eq: F budget constraint intermediate} and \eqref{eq: F budget constraint flexible}). Therefore, the set of distributions attainable in a given market environment depend strongly on the investor's initial wealth; varying the initial wealth generally results in a different set of attainable distributions.

\subsection{Conclusions and future directions}

Sharpe \emph{et al.} proposed the idea of having an expected utility maximizer choose a probability distribution for future wealth as an input to her investment problem instead of a utility function. The essence of their method is that an investor selects a desired probability distribution for future wealth and, subject to her initial wealth and market constraints, is then told the optimal policies and risk preferences consistent with that choice. We extended this normative approach to a continuous-time complete market framework with variable market coefficients. This results in added flexibility as to when the investor would like to realize her desired distribution as well as flexibility with the investment horizon itself.

Our method relies on being able to estimate the market price of risk, and one possible direction for future work is to address how to formulate and solve similar questions in a complete or incomplete market with stochastic market coefficients. We have also seen that the investor cannot arbitrarily choose multiple distributions for future wealth throughout the investment horizon in the model considered herein, regardless of whether she is a Merton investor or a forward investor with monontone performance criteria. Perhaps the selection of multiple distributions for future wealth can be done in a more general market model. Finally, another extension would be to consider a multi-period model, in the sense that the investor places a distribution for wealth at some future time $T_1$, invests optimally on $[0,T_1]$, and then at time $T_1$ selects another distribution for wealth to be placed at time $T_2>T_1$, having realized her wealth random variable at $T_1$ according to the previously chosen distribution. These are all subjects of ongoing research.

\section*{Acknowledgements}

The author would like to thank T. Zariphopoulou, M. S\^irbu, and G. \v Zitkovi\'c for their comments and suggestions.

\section*{Appendix: Proof of Proposition \ref{prop: classical utility heat equation} }

For completeness, we provide the proof of Proposition \ref{prop: classical utility heat equation}, which is an adaptation of the result of \citet{kal-2011} for the case of constant coefficients.

\begin{proof}
Under Assumptions \ref{assump: market price of risk} and \ref{assumption on utility}, it is well known (see, for example, \citet{kar-shr-1998}) that the optimal wealth process is given by $X_t^*=\psi(kZ_t,t)$, where the function $\psi\colon(0,\infty)\times[0,T]\to(0,\infty)$ is defined by
\begin{equation*}
 \psi(y,t) = \bbE_\bbQ\left[ I_T\left(y\frac{Z_T}{Z_t}\right)\right].
\end{equation*}
Herein, $\bbE_\bbQ$ denotes expectation under the equivalent martingale measure $\bbQ^T$ given by $\frac{d\bbQ^T}{d\bbP}=Z_T$ where $Z_T$ is as in \eqref{eq: state price density process}, while the Lagrange multiplier $k>0$ is the solution to
\begin{equation}\label{eq: budget}
 \bbE[Z_TI_T(kZ_T)]=x_0.
\end{equation}
Moreover, by the polynomial growth assumption \eqref{eq: I poly growth} on $I_T$ and the H\"{o}lder continuity of $|\lambda(t)|$, it is known (see \citet[Lemma 8.4 (p. 122)]{kar-shr-1998}) that $\psi\in C((0,\infty)\times[0,T])\cap C^{2,1}((0,\infty)\times[0,T))$ and solves the Cauchy problem
\begin{equation*}
 \left\{\begin{array}{ll} \psi_t(y,t)+\frac{1}{2}\abs{\lambda(t)}^2y^2\psi_{yy}(y,t)+\abs{\lambda(t)}^2y\psi_y(y,t)=0,& (y,t)\in (0,\infty)\times[0,T)\\
         \psi(y,T)=I_T(y),&y\in(0,\infty),
        \end{array}\right.
\end{equation*}
and that, for each $t\in[0,T)$, the function $y\mapsto \psi(y,t)$ is strictly decreasing.

Next, we define a function $h\colon\bbR\times[0,T]\to(0,\infty)$ by
\begin{equation*}
 h(x,t):= \psi(e^{-x+\frac{1}{2}A_t},t),
\end{equation*}
where $A_t$ is as in \eqref{eq: mkt inp proc}. Then 
\begin{equation}\label{eq: h orig tc}
h(x,T)=I_T\left( e^{-x+\frac{1}{2}A_T}\right). 
\end{equation}
Since the investor's optimal strategy is invariant under positive dilations of the argument of $I_T(\cdot)$ (by Remark \ref{rem: affine transformation}), we can assume the terminal condition is $h(x,T)=I_T(e^{-x})$.  We then have that $h\in C(\bbR\times[0,A_T])\cap C^{2,1}(\bbR\times[0,A_T))$ and solves
\begin{equation}\label{eq: classical utility h function}
 \left\{ \begin{array}{ll} h_t(x,t)+\frac{1}{2}|\lambda(t)|^2h_{xx}(x,t)=0,& (x,t)\in\bbR\times[0,T)\\
          h(x,T)=I_T\left( e^{-x}\right),& x\in\bbR.
         \end{array}\right.
\end{equation}
Let $h^{(-1)}$ denote the spatial inverse of $h$, which exists by the spatial monotonicity of $\psi$ and the relation $h_x(x,t) = -\psi_y(e^{-x+\frac{1}{2}A_t},t)e^{-x+\frac{1}{2}A_t}>0 ,\; (x,t)\in\bbR\times[0,T)$. Observe that by \eqref{eq: budget} we have $h\left(-\log(k),0\right)=\psi(k,0)=\bbE[Z_TI_T(kZ_T)] = x_0,$ and hence the underlying Lagrange multiplier satisfies
\begin{equation}\label{eq: hat y}
 k = e^{-h^{(-1)}(x_0,0)}.
\end{equation}
For $t\in[0,T]$, we then have
\begin{eqnarray}
 X_t^* &=& \psi(kZ_t,t)= \psi\left( e^{-h^{(-1)}(x_0,0)} e^{-M_t-\frac{1}{2}A_t},t\right)\nonumber\\
&=& \psi\left( e^{-( h^{(-1)}(x_0,0)+M_t+A_t)+\frac{1}{2}A_t},t\right)=h(h^{(-1)}(x_0,0)+M_t+A_t,t)\label{eq: classical utility opt wealth h},
\end{eqnarray}
and \eqref{eq: classical utility optimal wealth} follows.

Next, we recall the evolution of the optimal wealth process
\begin{equation}\label{eq: eu pf 1}
 dX_t^* = \sigma(t)\pi_t^*\cdot(\lambda(t)dt+dW_t),\qquad t\in[0,T].
\end{equation}
For $t\in[0,T]$, let $N_t:=h^{(-1)}(x_0,0)+M_t+A_t$ and observe that $N_t=h^{(-1)}\left(X_t^*,t\right),t\in[0,T],$ by \eqref{eq: classical utility optimal wealth}. By It\^{o}'s formula, the process $X_t^*, t\in[0,T]$, given in \eqref{eq: classical utility opt wealth h} satisfies
\begin{eqnarray}
 dX_t^*&=&  \left(h_t(N_t,t)+\frac{1}{2}\abs{\lambda(t)}^2h_{xx}(N_t,t)\right)dt + h_x(N_t,t)dN_t\nonumber\\
&=& \label{eq: eu pf 2} h_x\left( h^{(-1)}\left(X_t^*,t\right),t\right)\lambda(t)\cdot (\lambda(t)dt+dW_t).
\end{eqnarray}
Equating coefficients in \eqref{eq: eu pf 1} and \eqref{eq: eu pf 2}, we find that the optimal portfolio process $\pi_t^*$ is given by
\begin{equation*}
 \pi_t^* = h_x\left( h^{(-1)}\left(X_t^*,t\right),t\right)\sigma^{(-1)}(t)\lambda(t),\qquad t\in[0,T],
\end{equation*}
which yields the representation \eqref{eq: classical utility optimal portfolio} for the optimal portfolio process provided it is admissible. The admissibility is guaranteed by the polynomial growth assumption \eqref{eq: I poly growth} on $I_T$ and the uniform boundedness of $\lambda(t)$ on $[0,T]$ (see \citet[Theorem 3.5 (p. 93), and Remark 6.9(ii) (p. 97)]{kar-shr-1998}).
\end{proof}


\end{document}